\documentclass[english,3p]{elsarticle}
\usepackage{mathpazo}
\usepackage[T1]{fontenc}
\usepackage[latin9]{inputenc}
\usepackage{array}
\usepackage{booktabs}
\usepackage{multirow}
\usepackage{amsmath}
\usepackage{amsthm}
\usepackage{graphicx}
\usepackage{rotfloat}

\makeatletter

\providecommand{\tabularnewline}{\\}

\theoremstyle{plain}
\newtheorem{thm}{\protect\theoremname}
\theoremstyle{plain}
\newtheorem{prop}[thm]{\protect\propositionname}
\ifx\proof\undefined
\newenvironment{proof}[1][\protect\proofname]{\par
	\normalfont\topsep6\p@\@plus6\p@\relax
	\trivlist
	\itemindent\parindent
	\item[\hskip\labelsep\scshape #1]\ignorespaces
}{%
	\endtrivlist\@endpefalse
}
\providecommand{\proofname}{Proof}
\fi

\journal{arXiv}

\@ifundefined{showcaptionsetup}{}{%
 \PassOptionsToPackage{caption=false}{subfig}}
\usepackage{subfig}
\makeatother

\usepackage{babel}
\providecommand{\propositionname}{Proposition}
\providecommand{\theoremname}{Theorem}

\begin{document}

\begin{frontmatter}{}

\title{Extensive networks would eliminate the demand for pricing formulas}

\author[labela]{Jaegi Jeon}

\author[labela]{Kyunghyun Park}

\author[labelb]{Jeonggyu Huh\corref{cor1}}

\address[labela]{Department of Mathematical Sciences, Seoul National University,
Seoul 08826, Korea}

\address[labelb]{Department of Statistics, Chonnam National University, Gwangju 61186,
Korea}

\cortext[cor1]{corresponding author.\\
 E-mail address: huhjeonggyu@jnu.ac.kr}
\begin{abstract}
In this study, we generate a large number of implied volatilities for the Stochastic Alpha Beta Rho (SABR) model using a graphics processing unit (GPU) based simulation and enable an extensive neural network to learn them. This model does not have any exact pricing formulas for vanilla options, and neural networks have an outstanding ability to approximate various functions. Surprisingly, the network reduces the simulation noises by itself, thereby achieving as much accuracy as the Monte-Carlo simulation. Extremely high accuracy cannot be attained via existing approximate formulas. Moreover, the network is as efficient as the approaches based on the formulas. When evaluating based on high accuracy and efficiency, extensive networks can eliminate the necessity of the pricing formulas for the SABR model. Another significant contribution is that a novel method is proposed to examine the errors based on nonlinear regression. This approach is easily extendable to other pricing models for which it is hard to induce analytic formulas.
\end{abstract}

\end{frontmatter}{}

\section{Introduction}

Neural networks are often employed for regression because they have an
outstanding ability to approximate a wide range of functions (refer
to \citet{cybenko1989approximation} and \citet{hornik1989multilayer}
for the celebrated universal approximation theorem). Only a decade
ago, simple models, such as linear models, were preferred over neural networks. Most researchers used to believe that too many
parameters led to notorious overfitting, which has resulted in an artificial intelligence winter in the past. However, as numerous techniques have been invented to prevent such an occurrence, most researchers like to utilize such networks in their research. To review the literature concerning the application of networks in finance, refer to \citet{ruf2019neural} for network-based option pricing and \citet{henrique2019literature} for market predictions.

\begin{sidewaystable*}
\begin{centering}
\begin{tabular}{cccccc}
\toprule 
\multicolumn{2}{c}{} & Culkin (2017) & Brostrom (2018) & Ferguson (2018) & McGhee (2018)\tabularnewline
\midrule
\midrule 
\multicolumn{2}{c}{base model} & Black-Scholes & Black-Scholes & Black-Scholes & SABR ($\beta=1$)\tabularnewline
\midrule 
\multicolumn{2}{c}{options type} & vanilla & vanilla & basket & vanilla\tabularnewline
\midrule 
\multicolumn{2}{c}{pricing method} & closed formula & closed formula & MC simulation & finite difference method\tabularnewline
\midrule 
\multicolumn{2}{c}{network inputs} & $f_{0}/K$, $T$, $r$, $q$, $\bar{\alpha}$ & $f_{0}/K$, $T$, $r$, $\bar{\alpha}$ & $f_{0,i}/K$, $T$, $\bar{\alpha}_{i}$, $\rho_{ij}$ ($1\leq i,j\leq6$) & $f_{0}/K$, $T$, $\alpha_{0}$, $\nu$, $\rho$\tabularnewline
\midrule 
\multicolumn{2}{c}{network outputs} & $c/K$ & $c/K$ & $c/K$ & $\sigma^{I}$\tabularnewline
\midrule 
\multicolumn{2}{c}{\# of training samples} & 240k & 800k & 500M & 2.5M\tabularnewline
\multicolumn{2}{c}{(precision)} & (exact) & (exact) & (10k paths) & (\{$n_{t},n_{f},n_{\alpha}$\}=\{$400,200,100$\})\tabularnewline
\midrule 
\multicolumn{2}{c}{\# of test samples} & 60k & 200k & 5k & 500k\tabularnewline
\multicolumn{2}{c}{(precision)} & (exact) & (exact) & (100M paths) {[}?{]} & (\{$n_{t},n_{f},n_{\alpha}$\}=\{$400,200,100$\})\tabularnewline
\midrule
\multicolumn{2}{c}{\# of epochs} & 10 & 50 & 95 {[}?{]} & N/A\tabularnewline
\midrule 
\multicolumn{2}{c}{batch size} & 64 & 200 & 50k & N/A\tabularnewline
\midrule
\multicolumn{2}{c}{\# of hidden layers} & 4 & 2 & 6 & 1\tabularnewline
\midrule 
\multicolumn{2}{c}{\# of nodes per layer} & 100 & 256 & 1,400 & 1,000\tabularnewline
\midrule 
\multicolumn{2}{c}{activation functions} & ReLU, ELU, etc. & ReLU & ReLU & softplus, ReLU\tabularnewline
\midrule 
\multicolumn{2}{c}{optimizer} & SGD & ADAM & ADAM & ADAM\tabularnewline
\midrule 
\multicolumn{2}{c}{architecture tuning} & none & \# of layers, \# of nodes & \# of nodes & \# of nodes\tabularnewline
\midrule 
\multirow{2}{*}{test loss} & MSFE & 1.25E-4 & 7.27E-8 & 5E-5 {[}?{]} & N/A\tabularnewline
 & MSPE & same as above & same as above & N/A & N/A\tabularnewline
\bottomrule
\end{tabular}
\par\end{centering}
\centering{}\caption{\label{tab:liter1} This table outlines various methods to train neural networks using option prices $c$ (or implied volatilities $\sigma^{I}$). If there are several results in a work, only the best is written here.
The base models are expressed as follows: $df_{t}=\left(r-q\right)f_{t}dt+\bar{\alpha}f_{t}dW_{t}$
(Black-Scholes), $df_{t}=\left(r-q\right)f_{t}dt+\sqrt{y_{t}}f_{t}dW_{t}$,
$dy_{t}=\kappa\left(\bar{y}-y_{t}\right)+\nu\sqrt{y_{t}}dZ_{t}$ (Heston),
$df_{t}=\alpha_{t}f_{t}^{\beta}dW_{t}$, $d\alpha_{t}=\nu\alpha_{t}dZ_{t}$
(Stochastic Alpha Beta Rho [SABR]), where $dW_{t}dZ_{t}=\rho dt$. $K$ and $T$ are the strike and maturity of the vanilla option, respectively. The LR, SGD, MSFE, and MSPE stand for the learning rate, the stochastic gradient descent, the mean squared fitting error, and the mean squared prediction error, respectively (refer to Section \ref{sec:theorem} for the MSFE and MSPE). Finally, we used prefixes for the international system of
units: $1k=1000$ and $1M=1000k$. The mark {[}?{]} indicates that
the value is inferred using the context or the figures in the work.}
\end{sidewaystable*}
\begin{sidewaystable*}
\begin{centering}
\begin{tabular}{ccccccc}
\toprule 
\multicolumn{2}{c}{} & \multicolumn{2}{c}{Liu (2019)} & \multicolumn{2}{c}{Hirsa (2019)} & our method (2020)\tabularnewline
\midrule
\midrule 
\multicolumn{2}{c}{base model} & Black-Scholes & Heston & Black-Scholes & Heston & SABR ($\beta=1$)\tabularnewline
\midrule 
\multicolumn{2}{c}{options type} & vanilla & vanilla & vanilla & vanilla & vanilla\tabularnewline
\midrule 
\multicolumn{2}{c}{pricing method} & closed formula & Fourier-cosine series & closed formula & fast Fourier transform & MC simulation\tabularnewline
\midrule 
\multicolumn{2}{c}{network inputs} & $f_{0}/K$, $T$, $r$, $\bar{\alpha}$ & $f_{0}/K$, $T$, $r$, $y_{0}$, $\kappa$, $\bar{y}$, $\nu$, $\rho$ & $f_{0}/K$, $T$, $r$, $q$, $\bar{\alpha}$ & $f_{0}/K$, $T$, $r$, $q$, $y_{0}$, $\kappa$, $\bar{y}$, $\nu$,
$\rho$ & $K/f_{0}$, $T$, $\alpha_{0}$, $\nu$, $\rho$\tabularnewline
\midrule 
\multicolumn{2}{c}{network outputs} & $c/K$, $\sigma^{I}$ & $c$, $\sigma^{I}$ & $c/K$ & $c/K$ & $\sigma^{I}$\tabularnewline
\midrule 
\multicolumn{2}{c}{\# of training samples} & 900k & 900k (800k+100k) & 240k & 240k & 520M (480M+40M)\tabularnewline
\multicolumn{2}{c}{(precision)} & (exact) & (almost exact) & (exact) & (almost exact) & (500k paths)\tabularnewline
\midrule 
\multicolumn{2}{c}{\# of test samples} & 100k & 100k & 60k & 60k & 140M (40M+100M)\tabularnewline
\multicolumn{2}{c}{(precision)} & (exact) & (almost exact) & (exact) & (almost exact) & (500k, 12.5M paths)\tabularnewline
\midrule 
\multicolumn{2}{c}{\# of epochs} & \multicolumn{2}{c}{3000} & \multicolumn{2}{c}{N/A} & early stopping\tabularnewline
\midrule 
\multicolumn{2}{c}{batch size} & \multicolumn{2}{c}{1024} & \multicolumn{2}{c}{N/A} & 100\tabularnewline
\midrule 
\multicolumn{2}{c}{\# of hidden layers} & \multicolumn{2}{c}{4} & \multicolumn{2}{c}{4} & 2\tabularnewline
\midrule 
\multicolumn{2}{c}{\# of nodes per layer} & \multicolumn{2}{c}{400} & \multicolumn{2}{c}{120} & 7,000\tabularnewline
\midrule 
\multicolumn{2}{c}{activation functions} & \multicolumn{2}{c}{ReLU} & \multicolumn{2}{c}{ReLU, ELU, etc.} & ReLU\tabularnewline
\midrule 
\multicolumn{2}{c}{optimizer} & \multicolumn{2}{c}{LR decaying ADAM} & \multicolumn{2}{c}{ADAM} & LR decaying ADAM\tabularnewline
\midrule 
\multicolumn{2}{c}{architecture tuning} & \multicolumn{2}{c}{\# of nodes} & \multicolumn{2}{c}{\# of layers, \# of nodes} & \# of layers, \# of nodes\tabularnewline
\midrule 
\multirow{3}{*}{test loss} & \multirow{2}{*}{MSFE} & 8.21E-9 ($c/K$) & 1.65E-8 ($c$) & 1.8E-5 {[}?{]} & 2.5E-5 {[}?{]} & 5.5E-6\tabularnewline
 &  & 1.55E-8 ($\sigma^{I}$) & 5.07E-7 ($\sigma^{I}$) &  &  & \tabularnewline
 & MSPE & same as above & same as above & same as above & same as above & 2.0E-7\tabularnewline
\bottomrule
\end{tabular}
\par\end{centering}
\centering{}\caption{\label{tab:liter2}This table summarizes methods to train neural networks using option prices $c$ (or implied volatilities $\sigma^{I}$). If there are several results in a work, only the best is written here.
The base models are expressed as follows: $df_{t}=\left(r-q\right)f_{t}dt+\bar{\alpha}f_{t}dW_{t}$
(Black-Scholes), $df_{t}=\left(r-q\right)f_{t}dt+\sqrt{y_{t}}f_{t}dW_{t}$,
$dy_{t}=\kappa\left(\bar{y}-y_{t}\right)+\nu\sqrt{y_{t}}dZ_{t}$ (Heston),
$df_{t}=\alpha_{t}f_{t}^{\beta}dW_{t}$, $d\alpha_{t}=\nu\alpha_{t}dZ_{t}$
(SABR), where $dW_{t}dZ_{t}=\rho dt$. In particular, $K$ and $T$ are the strike and maturity of the vanilla option, respectively. The LR, SGD, MSFE,
and MSPE stand for a learning rate, the stochastic gradient descent,
the mean squared fitting error, and the mean squared prediction error,
respectively (refer to Section \ref{sec:theorem} for the MSFE and
MSPE). Finally, we used prefixes for the international system of
units: $1k=1000$ and $1M=1000k$. The mark {[}?{]} signifies that
the value is deduced using the context or the figures in the work.}
\end{sidewaystable*}

Since the work of \citet{hutchinson1994nonparametric}, many researchers
have been studying artificial neural networks to predict the
option prices $c$ (or the implied volatilities $\sigma^{I}$) for
particular parametric models, such as the Black-Scholes model \citep{black1973pricing},
the Heston model \citep{heston1993closed}, and the SABR model \citep{hagan2002managing}.
We selected six related studies \citep{culkin2017machine,brostrom2018exotic,ferguson2018deeply,mcghee2018artificial,liu2019pricing,hirsa2019supervised}
and summarized their approaches in Tables \ref{tab:liter1} and
\ref{tab:liter2}. Interestingly, we can observe similar propensities
in the studies. They mainly focus on the models that can
provide efficient pricing formulas for vanilla options such as the
Black-Scholes model and the Heston model. For instance, the Heston
model gives a closed-form characteristic function, enabling cost-effective
option pricing through a Fourier transform \citep{rouah2013heston}.
Note that all works other than \citet{mcghee2018artificial} are either associated with the Black-Scholes model or the Heston model. This correlation may exist because training samples are generated exhaustively to eliminate the necessity of numerical algorithms. We, however, think that their contributions are a little marginal from a practical perspective as even without the networks, the option prices can already be efficiently obtained.

Notably, \citet{mcghee2018artificial} aimed to allow a neural
network to learn the vanilla option prices for the SABR model. The model
does not offer exact and efficient solutions for a true option
value $c_{true}$. Thus, a finite difference method (FDM) for second-order in space and first-order in time was utilized to produce an approximation $c_{approx}$ of $c_{true}$. Consequently, the network proposed by the study produces outcomes much more quickly than the FDM and outperforms the well-known approximation of \citet{hagan2002managing}
in terms of accuracy. Nonetheless, \citeauthor{mcghee2018artificial} only tries
to make the predicted value $c_{net}$ of the network come close to $c_{approx}$.
In essence, the prediction error $c_{net}-c_{true}$ is not considered
in the study, and only the reduction of the fitting error $c_{net}-c_{approx}$ is studied. Nevertheless, $c_{net}$ should be close
to $c_{true}$ (i.e., not to $c_{approx}$). Moreover, the research is not
conducted systematically. We could not find any mentions about the
number of epochs, the batch size, the weight initialization method,
and the loss values for the training and test datasets.
In particular, the types of neural networks tested in the research
are fairly limited because the number of the hidden layers for the
networks is fixed at one.

We believe that it is desirable to choose a parametric model without
an exact pricing formula for vanilla options and investigate a training
method of neural networks using numerous option prices for the model.
Therefore, we decided to study the SABR model. Furthermore, a pricing method to generate big data should be efficient enough and easily applicable to a wide range of models, such as the rough volatility model \citep{gatheral2018volatility}.
Standard procedures satisfying the requirements may be the FDM and
Monte-Carlo simulation (MC). Particularly, when the number
of factors for underlying models is smaller than four, the FDM is
usually more efficient and stable than the MC (see \citet{wilmott2013paul}).
This evidence denotes that the FDM may be more appropriate to price vanilla options under the SABR model than the MC because the model has two factors.
Nevertheless, we chose the MC to produce $c_{approx}$ because we
think that neural networks have the potential to filter out symmetric
noises caused by the MC, but they are unable to rectify the bias caused
by the FDM. In other words, we do not choose the FDM but rather the MC because $E\left[c_{approx}\right]=c_{true}$ for the MC case, but $E\left[c_{approx}\right]\neq c_{true}$ for the FDM case. For example, \citet{ferguson2018deeply} made two datasets using
the MC, where one was precise but small, and the other was big yet imprecise, and they trained two networks with the datasets, respectively. Interestingly, the network using the larger and less precise dataset gives better results. This outcome indicates that the network can reduce the MC noises due to the integration of larger data. Thus, we also expect that the MC errors are reduced by a neural network provided the training data are sufficiently large. 

Furthermore, based on nonlinear regression analysis, we propose a
novel method to indirectly estimate the prediction error $c_{net}-c_{true}$
even if $c_{true}$ cannot be obtained. In other studies,
the prediction error is usually expected to be greater than the approximation error $c_{approx}-c_{true}$ because it is a sum of the approximation error and the fitting error $c_{net}-c_{true}$ that arises during training. However, neural networks are able
to reduce a substantial part of the approximation error and produce
$c_{net}$ close to $c_{true}$. The method developed in this study
would be an invaluable tool to evaluate the degree of the distance
between $c_{net}$ and $c_{true}$. When analyzing test results with
the approach, the accuracy of our network is estimated to be comparable
to that of about 13 million MC simulations.

Further, it is noticeable in the literature that although similar
approaches are adopted for the same model, the loss of test
data differ considerably depending on the details of the training
methods. This notion is validated via the mean squared fitting errors (MSFE) of \citet{culkin2017machine}, \citet{brostrom2018exotic}, \citet{liu2019pricing}, and \citet{hirsa2019supervised} for the Black-Scholes model in Table \ref{tab:liter1} and \ref{tab:liter2}. The MSFEs of \citeauthor{brostrom2018exotic} and \citeauthor{liu2019pricing} ($7.27\times10^{-8}$ and $8.21\times10^{-9}$, respectively) seem to be superior to the MSFEs of \citeauthor{culkin2017machine}
and \citeauthor{hirsa2019supervised} ($1.25\times10^{-4}$ and $1.85\times10^{-5}$, respectively). This association may be because the latter networks have narrower structures or learn from smaller data than the former. The former have 400 nodes per layer, and the latter have 100 or 140 nodes. Furthermore, the training data sizes of the former are 800 or 900 thousands, while the latter is 240 thousands. Either or both of the two options can facilitate performance gaps. These gaps can also be confirmed through the MSFEs for the Heston model of \citeauthor{liu2019pricing} and \citeauthor{hirsa2019supervised} ($1.65\times10^{-8}$ and $2.5\times10^{-5}$, respectively). This notion implies that it is not easy to perfectly fit neural networks to the generated data. For a better goodness-of-fit, one should consider several factors, such as data size, network architecture, and the optimization method. Therefore, we try to reduce the MSFE in our experiment by generating enormous data and tuning various hyperparameters for a considerably accurate fit. 

In summary, this study contributes to the literature in the following
ways. First, we generate numerous data using GPU-based simulations
and verify that the network trained with the data provides extraordinarily
accurate results. Second, we provide a novel method to analyze the
prediction errors $c_{net}-c_{true}$ by proposing an unbiased and
consistent estimator associated with the prediction error. We make
a new attempt using nonlinear regression, which forms the theoretical
basis for the phenomenon that the network produces prediction
errors smaller than the approximation errors $c_{approx}-c_{true}$.

The remainder of this paper is organized as follows. In Section 2, a new
method of analyzing prediction errors with nonlinear regression is
introduced. We then discuss the pricing methods of options in
the SABR model and detailed methods of generating data for network
learning in Section 3. In Section 4, we train neural networks of various
structures and assess the impact of training data size on network
performance. Finally, Section 5 concludes the study.

\section{\label{sec:theorem}Nonlinear regression of numerous implied volatilities}

As mentioned in the introduction, we only focus on the parametric
models that do not have any exact pricing formulas. Therefore, an
important step in this work is to generate numerous approximate implied
volatilities $\sigma_{approx,l}^{I}$ for exact volatilities $\sigma_{true,l}^{I}$($l=1,2,\cdots,L$) under a parametric model, which will be used as the material to train networks. Each volatility $\sigma_{approx,l}^{I}$ is generated on randomly chosen parameters, such as $\theta_{l,1}$, $\theta_{l,2}$, $\cdots$, $\theta_{l,n_{\theta}}$, maturity $T_{l}$, and strike $K_{l}$.

As $\sigma_{true,l}^{I}$ is determined by $\theta_{l,1}$, $\theta_{l,2}$,
$\cdots$, $\theta_{l,n_{\theta}}$, $T_{l}$, and $K_{l}$, there exists
a function $\tilde{h}$ such that $\sigma_{true,l}^{I}=\tilde{h}\left(x_{l}\right)$
for $x_{l}=\left(1,\theta_{l,1},\theta_{l,2},\cdots,\theta_{l,n_{\theta}},T_{l},K_{l}\right)$.
According to the renowned universal approximation theorem, it is
assumed that a network with enough number of weights $\Gamma=\{\gamma_{1},\gamma_{2},\cdots,\gamma_{n_{\gamma}}\}$
can accurately approximate the function value $\tilde{h}\left(x_{l}\right)$
as its output $h\left(x_{l};\Gamma\right)$. Thus, if $M$ simulations
are run to obtain the approximate volatility $\sigma_{approx,l}^{I}$,
the central limit theorem yields the following relation:
\begin{align*}
\sigma_{approx,l}^{I}\left(M\right) & =h\left(x_{l};\Gamma\right)+\epsilon_{l}\left(M\right),
\end{align*}
where $\epsilon_{l}\left(M\right)\sim N\left(0,\beta_{l}^{2}/M\right)$
for $\beta_{l}>0$ (\citet{glasserman2013monte}). Here, the time interval
for the simulations is supposed to be small enough to neglect the
bias of $\sigma_{approx,l}^{I}$ against $\sigma_{true,l}^{I}$, and
the notation $\omega\left(M\right)$ is applied to emphasize that
$\omega$ is a random variable that is dependent on $M$. 

In the perspective of nonlinear ordinary regression (\citet{montgomery2012introduction}),
an unbiased and consistent estimator $\hat{\Gamma}$ of $\Gamma$
is 

\[
\hat{\Gamma}=\underset{\Gamma}{{\rm argmin}}\mathcal{L}\left(\Gamma;M\right),
\]
where
\[
\mathcal{L}\left(\Gamma;M\right)=\frac{1}{2}\sum_{l=1}^{L}\left(h\left(x_{l};\Gamma\right)-\sigma_{approx,l}^{I}\left(M\right)\right)^{2}.
\]
The Jacobian and Hessian matrices $\boldsymbol{J}$ and $\boldsymbol{H}$
of $\mathcal{L}\left(\hat{\Gamma};M\right)$ should satisfy the optimality
condition that $\boldsymbol{J}$ and $\boldsymbol{H}$ are zero and
positive definite, respectively. On the other hand, we can derive
\begin{equation}
\boldsymbol{J}=\boldsymbol{\epsilon}\boldsymbol{Q},\quad{\rm \boldsymbol{H}\approx\boldsymbol{Q}^{T}\boldsymbol{Q}},\label{eq:jaco_hess}
\end{equation}
(\citet{hansen2013least}), where 
\[
\boldsymbol{\epsilon}=\left[\begin{array}{cccc}
\epsilon_{1} & \epsilon_{2} & \cdots & \epsilon_{L}\end{array}\right],\quad\boldsymbol{Q}=\left[\begin{array}{cccc}
\frac{\partial h\left(x_{1};\hat{\Gamma}\right)}{\partial\gamma_{1}} & \frac{\partial h\left(x_{1};\hat{\Gamma}\right)}{\partial\gamma_{2}} & \cdots & \frac{\partial h\left(x_{1};\hat{\Gamma}\right)}{\partial\gamma_{n_{\gamma}}}\\
\frac{\partial h\left(x_{2};\hat{\Gamma}\right)}{\partial\gamma_{1}} & \frac{\partial h\left(x_{2};\hat{\Gamma}\right)}{\partial\gamma_{2}} & \cdots & \frac{\partial h\left(x_{2};\hat{\Gamma}\right)}{\partial\gamma_{n_{\gamma}}}\\
\vdots & \vdots & \ddots & \vdots\\
\frac{\partial h\left(x_{L};\hat{\Gamma}\right)}{\partial\gamma_{1}} & \frac{\partial h\left(x_{L};\hat{\Gamma}\right)}{\partial\gamma_{2}} & \cdots & \frac{\partial h\left(x_{L};\hat{\Gamma}\right)}{\partial\gamma_{n_{\gamma}}}
\end{array}\right].
\]
Furthermore, $\hat{\Gamma}$ follows a multivariate normal distribution
as follows: 
\[
\hat{\Gamma}\sim N\left(\Gamma,\frac{1}{LM}\boldsymbol{W}^{-1}\boldsymbol{W}^{\beta}\boldsymbol{W}^{-1}\right),
\]
where $\boldsymbol{W}=\frac{1}{L}\boldsymbol{Q}^{T}\boldsymbol{Q}$,
$\boldsymbol{W}^{\beta}=\frac{1}{L}\boldsymbol{Q}^{T}\boldsymbol{B}\boldsymbol{Q}$,
and $\boldsymbol{B}$ is the diagonal matrix with the $l$th diagonal
entry $\beta_{l}^{2}$. Conversely, owing to the law of large
numbers, the elements of $\boldsymbol{W}$ and $\boldsymbol{W}_{\beta}$
converge in probability to their respective expected values as $L\rightarrow\infty$.
In other words, 
\begin{gather*}
\boldsymbol{W}_{k,k'}=\left\langle \frac{\partial h\left(x_{l};\hat{\Gamma}_{M}\right)}{\partial\gamma_{k}}\frac{\partial h\left(x_{l};\hat{\Gamma}_{M}\right)}{\partial\gamma_{k'}}\right\rangle _{l,L}\overset{p}{\rightarrow}\left\langle \frac{\partial h\left(x_{l};\hat{\Gamma}_{M}\right)}{\partial\gamma_{k}}\frac{\partial h\left(x_{l};\hat{\Gamma}_{M}\right)}{\partial\gamma_{k'}}\right\rangle _{l},\\
\boldsymbol{W}_{k,k'}^{\beta}=\left\langle \beta_{l}^{2}\frac{\partial h\left(x_{l};\hat{\Gamma}_{M}\right)}{\partial\gamma_{k}}\frac{\partial h\left(x_{l};\hat{\Gamma}_{M}\right)}{\partial\gamma_{k'}}\right\rangle _{l,L}\overset{p}{\rightarrow}\left\langle \beta_{l}^{2}\frac{\partial h\left(x_{l};\hat{\Gamma}_{M}\right)}{\partial\gamma_{k}}\frac{\partial h\left(x_{l};\hat{\Gamma}_{M}\right)}{\partial\gamma_{k'}}\right\rangle _{l}
\end{gather*}
as $L\rightarrow\infty$, where $\left\langle \varphi_{l}\right\rangle _{l,L}$
and $\left\langle \varphi_{l}\right\rangle _{l}$ are a sample mean
of size $L$ and the population mean of a random variable $\varphi_{l}$,
respectively. That is, $\left\langle \varphi_{l}\right\rangle _{l,L}=\frac{1}{L}\sum_{l=1}^{L}\varphi_{l}$,
and $\left\langle \varphi_{l}\right\rangle _{l}=E\left[\varphi_{l}\right]$.
Notably, $E\left[\left\langle \varphi_{l}\right\rangle _{l,L}\right]=\left\langle \varphi_{l}\right\rangle _{l}$.
Accordingly , if $L\gg1$, then $\boldsymbol{W}$ and $\boldsymbol{W}^{\beta}$
rarely change, although $L$ does change a little.

We now consider another dataset for an out-of-sample test, constituting
$L'$ volatilities $\sigma_{approx,l}^{I}$ to approximate $\sigma_{true,l}^{I}$
($l=1,2,\cdots,L$'), each of which is generated from $M'$ simulations.
With regard to the dataset, the fitting error $\epsilon_{fit,l}$,
the prediction error $\epsilon_{pred,l}$, and the approximation error
$\epsilon_{approx,l}$ are defined as follows:
\begin{gather*}
\epsilon_{fit,l}\left(\hat{\Gamma},M'\right)=\sigma_{net,l}^{I}\left(\hat{\Gamma}\right)-\sigma_{approx,l}^{I}\left(M'\right),\\
\epsilon_{pred,l}\left(\hat{\Gamma}\right)=\sigma_{net,l}^{I}\left(\hat{\Gamma}\right)-\sigma_{true,l}^{I},\quad\epsilon_{approx,l}\left(M'\right)=\sigma_{approx,l}^{I}\left(M'\right)-\sigma_{true,l}^{I},
\end{gather*}
where $\sigma_{net,l}^{I}\left(\hat{\Gamma}\right)=h\left(x_{l};\hat{\Gamma}\right)$.
Notably, $\epsilon_{pred,l}$ can be decomposed into $\epsilon_{fit,l}$
and $\epsilon_{approx,l}$, that is, 
\[
\epsilon_{pred,l}\left(\hat{\Gamma}\right)=\epsilon_{fit,l}\left(\hat{\Gamma},M'\right)+\epsilon_{approx,l}\left(M'\right).
\]
Note that finding $\epsilon_{fit,l}$ is straightforward while determining
$\epsilon_{pred,l}$ and $\epsilon_{approx,l}$ is not simple because
$\sigma_{true,l}$ is unknown. Many researchers intuitively expect
$|\epsilon_{pred,l}|>|\epsilon_{approx,l}|$ because they guess that the
signs of $\epsilon_{fit,l}$ and $\epsilon_{approx,l}$ are the same.
Nonetheless, it would be the best if $\epsilon_{fit,l}$ canceled out
a part of $\text{\ensuremath{\epsilon}}_{approx,l}$ so that $|\epsilon_{pred,l}|<|\epsilon_{approx,l}|$.
This mechanism is possible only when the neural network can
reduce the noises in $\epsilon_{approx,l}$ and find more plausible
values by itself. Specifically, we prove that self-correction of networks
is feasible, and it will be demonstrated in the tests of Section
\ref{sec:test}. 

Additionally, the errors $\epsilon_{fit,l}$, $\epsilon_{pred,l}$,
and $\epsilon_{approx,l}$ follow their respective normal distributions
as below: 
\begin{gather}
\epsilon_{fit,l}\left(\hat{\Gamma},M'\right)\sim N\left(0,\frac{\beta_{l}^{2}}{M'}+\frac{1}{LM}\boldsymbol{q_{l}}\boldsymbol{W}^{-1}\boldsymbol{W}^{\beta}\boldsymbol{W}^{-1}\boldsymbol{q_{l}^{T}}\right),\nonumber \\
\epsilon_{approx,l}\left(M'\right)\sim N\left(0,\frac{\beta_{l}^{2}}{M'}\right),\quad\epsilon_{pred,l}\left(\hat{\Gamma}\right)\sim N\left(0,\frac{1}{LM}\boldsymbol{q_{l}}\boldsymbol{W}^{-1}\boldsymbol{W}^{\beta}\boldsymbol{W}^{-1}\boldsymbol{q_{l}^{T}}\right),\label{eq:err_dist}
\end{gather}
where $\boldsymbol{q_{l}}$ is the $l$th row vector of $\boldsymbol{Q}$.
As mentioned above, it is extremely important to note that finding $\epsilon_{pred,l}$
is infeasible because $\sigma_{true,l}^{I}$ is unknown. This problem is serious because we need $\epsilon_{pred,l}$ to evaluate
the performance of the network. Although some might presume that the difficulty can be circumvented by computing $\boldsymbol{Q}$, $\boldsymbol{W}$, and $\boldsymbol{W}_{\beta}$, the computation is severely unstable due to the countless parameters of the network. 

To resolve the problem, we define three mean squared errors (MSE)
for the test dataset, namely the mean squared fitting error $\mathcal{E}_{fit}$(MSFE), the mean squared prediction error $\mathcal{E}_{pred}$ (MSPE), and the mean squared approximation error $\mathcal{E}_{approx}$ (MSAE).
They are given by the following:
\begin{gather*}
\mathcal{E}_{fit}\left(\hat{\Gamma},M';L'\right)=\frac{1}{L'}\sum_{l=1}^{L'}\epsilon_{fit,l}^{2}\left(\hat{\Gamma},M'\right),\\
\mathcal{E}_{pred}\left(\hat{\Gamma};L'\right)=\frac{1}{L'}\sum_{l=1}^{L'}\epsilon_{pred,l}^{2}\left(\hat{\Gamma}\right),\quad\mathcal{E}_{approx}\left(M';L'\right)=\frac{1}{L'}\sum_{l=1}^{L'}\epsilon_{approx,l}^{2}\left(M'\right).
\end{gather*}
Among them, the MSPE $\mathcal{E}_{pred}$ can serve as an indicator depicting the performance of the network with the weight
$\hat{\Gamma}$. However, it also depends on the type of test set.
Therefore, the following statistic will be utilized as an indicator
to gauge performance: 
\[
\mathcal{E}_{pred}\left(\hat{\Gamma}\right)=E\left[\mathcal{E}_{pred}\left(\hat{\Gamma};L'\right)\right],
\]
which is the same as $\left\langle \epsilon_{pred,l}^{2}\left(\hat{\Gamma}\right)\right\rangle _{l}$
because $\text{\ensuremath{\mathcal{E}_{pred}\left(\hat{\Gamma};L'\right)}=}\left\langle \epsilon_{pred,l}^{2}\left(\hat{\Gamma}\right)\right\rangle _{l,L'}$.
We will explain the estimation of $\mathcal{E}_{pred}\left(\hat{\Gamma}\right)$
in the later sections. The propositions below describe the expectations and variances
of $\mathcal{E}_{fit}$, $\mathcal{E}_{pred}$, and $\mathcal{E}_{approx}$.
\begin{prop}
\label{prop:expectation_of_error}The expectations of $\mathcal{E}_{fit}$,
$\mathcal{E}_{pred}$, and $\mathcal{E}_{approx}$ are given by 
\begin{gather*}
E\left[\mathcal{E}_{fit}\left(\hat{\Gamma},M';L'\right)\right]=E\left[\mathcal{E}_{approx}\left(M';L'\right)\right]+E\left[\mathcal{E}_{pred}\left(\hat{\Gamma};L'\right)\right],\\
E\left[\mathcal{E}_{approx}\left(M';L'\right)\right]=\frac{1}{M'}\left\langle \beta_{l}^{2}\right\rangle _{l,L'},\quad E\left[\mathcal{E}_{pred}\left(\hat{\Gamma};L'\right)\right]=\frac{1}{LM}\left\langle \boldsymbol{q_{l}}\boldsymbol{W}^{-1}\boldsymbol{W}^{\beta}\boldsymbol{W}^{-1}\boldsymbol{q_{l}^{T}}\right\rangle _{l,L'}.
\end{gather*}
\end{prop}
\begin{proof}
As $E\left[\epsilon_{fit,l}\right]=0$ for all $l$,
\begin{align*}
E\left[\mathcal{E}_{fit}\left(\hat{\Gamma},M';L'\right)\right] & =\frac{1}{L'}\sum_{l=1}^{L'}Var\left[\epsilon_{fit,l}\left(\hat{\Gamma},M'\right)\right]\\
 & =\frac{1}{L'}\sum_{l=1}^{L'}\left(\frac{\beta_{l}^{2}}{M'}+\frac{1}{LM}\boldsymbol{q_{l}}\boldsymbol{W}^{-1}\boldsymbol{W}^{\beta}\boldsymbol{W}^{-1}\boldsymbol{q_{l}^{T}}\right)\\
 & =\frac{1}{M'}\left\langle \beta_{l}^{2}\right\rangle _{l,L'}+\frac{1}{LM}\left\langle \boldsymbol{q_{l}}\boldsymbol{W}^{-1}\boldsymbol{W}^{\beta}\boldsymbol{W}^{-1}\boldsymbol{q_{l}^{T}}\right\rangle _{l,L'}.
\end{align*}
Similarly, $E\left[\mathcal{E}_{approx}\right]$ and $E\left[\mathcal{E}_{pred}\right]$
are calculated as $\frac{1}{M'}\left\langle \beta_{l}^{2}\right\rangle _{l,L'}$
and $\frac{1}{LM}\left\langle \boldsymbol{q_{l}}\boldsymbol{W}^{-1}\boldsymbol{W}^{\beta}\boldsymbol{W}^{-1}\boldsymbol{q_{l}^{T}}\right\rangle _{l,L'}$,
respectively. 
\end{proof}
\begin{prop}
\label{prop:variance_of_error}The variances of $\mathcal{E}_{fit}$,
$\mathcal{E}_{pred}$, and $\mathcal{E}_{approx}$ are
\begin{align*}
Var\left[\mathcal{E}_{fit}\left(\hat{\Gamma},M';L'\right)\right] & =Var\left[\mathcal{E}_{approx}\left(M';L'\right)\right]+Var\left[\mathcal{E}_{pred}\left(\hat{\Gamma};L'\right)\right]+\frac{4}{L'LM'M}\left\langle \beta_{l}^{2}\boldsymbol{q_{l}}\boldsymbol{W}^{-1}\boldsymbol{W}^{\beta}\boldsymbol{W}^{-1}\boldsymbol{q_{l}^{T}}\right\rangle _{l,L'},
\end{align*}
\[
Var\left[\mathcal{E}_{approx}\left(M';L'\right)\right]=\frac{2}{L'\left(M'\right)^{2}}\left\langle \beta_{l}^{4}\right\rangle _{l,L'},\quad Var\left[\mathcal{E}_{pred}\left(\hat{\Gamma};L'\right)\right]=\frac{2}{L'L^{2}M^{2}}\left\langle \left(\boldsymbol{q_{l}}\boldsymbol{W}^{-1}\boldsymbol{W}^{\beta}\boldsymbol{W}^{-1}\boldsymbol{q_{l}^{T}}\right)^{2}\right\rangle _{l,L'}.
\]
\end{prop}
\begin{proof}
The square $X^{2}$ of a normal random variable $X\sim N\left(0,\sigma^{2}\right)$
follows a gamma distribution $\Gamma\left(1/2,2\sigma^{2}\right)$,
which leads to
\[
\epsilon_{fit,l}^{2}\left(\hat{\Gamma},M'\right)\sim\Gamma\left(\frac{1}{2},2\left(\frac{\beta_{l}^{2}}{M'}+\frac{1}{LM}\boldsymbol{q_{l}}\boldsymbol{W}^{-1}\boldsymbol{W}^{\beta}\boldsymbol{W}^{-1}\boldsymbol{q_{l}^{T}}\right)\right).
\]
As $E\left[Y\right]=ab$ and $Var\left[Y\right]=ab^{2}$ for
$Y\sim\Gamma\left(a,b\right)$, $E\left[\epsilon_{fit,l}^{2}\right]=\frac{\beta_{l}^{2}}{M'}+\frac{1}{LM}\boldsymbol{q_{l}}\boldsymbol{W}^{-1}\boldsymbol{W}^{\beta}\boldsymbol{W}^{-1}\boldsymbol{q_{l}^{T}}$,
and $Var\left[\epsilon_{fit,l}^{2}\right]=2\left(\frac{\beta_{l}^{2}}{M'}+\frac{1}{LM}\boldsymbol{q_{l}}\boldsymbol{W}^{-1}\boldsymbol{W}^{\beta}\boldsymbol{W}^{-1}\boldsymbol{q_{l}^{T}}\right)^{2}$.
As $\epsilon_{fit,l}$ are independent, 
\begin{alignat*}{1}
 & Var\left[\mathcal{E}_{fit}\left(\hat{\Gamma},M';L'\right)\right]\\
 & =\frac{1}{\left(L'\right)^{2}}\sum_{l=1}^{L'}Var\left[\epsilon_{fit,l}^{2}\left(\hat{\Gamma}_{M},M'\right)\right]\\
 & =\frac{1}{\left(L'\right)^{2}}\sum_{l=1}^{L'}\left(\frac{2\beta_{l}^{4}}{\left(M'\right)^{2}}+\frac{2}{L^{2}M^{2}}\left(\boldsymbol{q_{l}}\boldsymbol{W}^{-1}\boldsymbol{W}^{\beta}\boldsymbol{W}^{-1}\boldsymbol{q_{l}^{T}}\right)^{2}+\frac{4\beta_{l}^{2}}{LM'M}\boldsymbol{q_{l}}\boldsymbol{W}^{-1}\boldsymbol{W}^{\beta}\boldsymbol{W}^{-1}\boldsymbol{q_{l}^{T}}\right)\\
 & =\frac{2}{L'\left(M'\right)^{2}}\left\langle \beta_{l}^{4}\right\rangle _{l,L'}+\frac{2}{L'L^{2}M^{2}}\left\langle \left(\boldsymbol{q_{l}}\boldsymbol{W}^{-1}\boldsymbol{W}^{\beta}\boldsymbol{W}^{-1}\boldsymbol{q_{l}^{T}}\right)^{2}\right\rangle _{l,L'}+\frac{4}{L'LM'M}\left\langle \beta_{l}^{2}\boldsymbol{q_{l}}\boldsymbol{W}^{-1}\boldsymbol{W}^{\beta}\boldsymbol{W}^{-1}\boldsymbol{q_{l}^{T}}\right\rangle _{l,L'}.
\end{alignat*}
Similarly, $Var\left[\mathcal{E}_{approx}\right]$ and $Var\left[\mathcal{E}_{pred}\right]$
are derived as $\frac{2}{L'\left(M'\right)^{2}}\left\langle \beta_{l}^{4}\right\rangle _{l,L'}$
and $\frac{2}{L'L^{2}M^{2}}\left\langle \left(\boldsymbol{q_{l}}\boldsymbol{W}^{-1}\boldsymbol{W}^{\beta}\boldsymbol{W}^{-1}\boldsymbol{q_{l}^{T}}\right)^{2}\right\rangle _{l,L'}$,
respectively. 
\end{proof}
Based on the propositions, the following theorem suggests an unbiased
and consistent estimator of $\mathcal{E}_{pred}\left(\hat{\Gamma}\right)$.
The theorem needs two distinct test sets.
\begin{thm}
\label{thm:main_theorem}The estimator 
\[
\hat{\mathcal{E}}_{pred}\left(\hat{\Gamma}\right)=\frac{M_{1}'\mathcal{E}_{fit}\left(M_{1}';L_{1}'\right)-M_{2}'\mathcal{E}_{fit}\left(M_{2}';L_{2}'\right)}{M_{1}'-M_{2}'}
\]
is unbiased and consistent to $\mathcal{E}_{pred}\left(\hat{\Gamma}\right)$
for $M_{1}'\neq M_{2}'$. Particularly, $\left(L_{1}',L_{2}'\right)$ and
$\left(M_{1}',M_{2}'\right)$ are the data lengths and the numbers
of simulations for two distinct test sets, respectively. Further,
the variance of $\hat{\mathcal{E}}_{pred}\left(\hat{\Gamma}\right)$
is given by
\begin{equation}
Var\left[\hat{\mathcal{E}}_{pred}\left(\hat{\Gamma}\right)\right]=\left(\frac{M_{1}'}{M_{1}'-M_{2}'}\right)^{2}Var\left[\mathcal{E}_{fit}\left(\hat{\Gamma}_{M},M_{1}';L_{1}'\right)\right]+\left(\frac{M_{2}'}{M_{1}'-M_{2}'}\right)^{2}Var\left[\mathcal{E}_{fit}\left(\hat{\Gamma}_{M},M_{2}';L_{2}'\right)\right],\label{eq:var_pred}
\end{equation}
where
\[
Var\left[\mathcal{E}_{fit}\left(\hat{\Gamma},M';L'\right)\right]=\frac{2\left\langle \beta_{l}^{4}\right\rangle _{l,L'}}{L'\left(M'\right)^{2}}+\frac{2\left\langle \left(\boldsymbol{q_{l}}\boldsymbol{W}^{-1}\boldsymbol{W}^{\beta}\boldsymbol{W}^{-1}\boldsymbol{q_{l}^{T}}\right)^{2}\right\rangle _{l,L'}}{L'L^{2}M^{2}}+\frac{4\left\langle \beta_{l}^{2}\boldsymbol{q_{l}}\boldsymbol{W}^{-1}\boldsymbol{W}^{\beta}\boldsymbol{W}^{-1}\boldsymbol{q_{l}^{T}}\right\rangle _{l,L'}}{L'LM'M}.
\]
\end{thm}
\begin{proof}
The unbiasedness of $\hat{\mathcal{E}}_{pred}\left(\hat{\Gamma}\right)$
in relation to $\mathcal{E}_{pred}\left(\hat{\Gamma}\right)$ is exhibited as follows:
\begin{align*}
E\left[\hat{\mathcal{E}}_{pred}\left(\hat{\Gamma}\right)\right] & =\frac{M_{1}'E\left[\mathcal{E}_{fit}\left(M_{1}';L_{1}'\right)\right]-M_{2}'E\left[\mathcal{E}_{fit}\left(M_{2}';L_{2}'\right)\right]}{M_{1}'-M_{2}'}\\
 & =\frac{\left(\left\langle \beta_{l}^{2}\right\rangle _{l}+\frac{M_{1}'}{LM}\left\langle \boldsymbol{q_{l}}\boldsymbol{W}^{-1}\boldsymbol{W}^{\beta}\boldsymbol{W}^{-1}\boldsymbol{q_{l}^{T}}\right\rangle _{l}\right)-\left(\left\langle \beta_{l}^{2}\right\rangle _{l}+\frac{M_{2}'}{LM}\left\langle \boldsymbol{q_{l}}\boldsymbol{W}^{-1}\boldsymbol{W}^{\beta}\boldsymbol{W}^{-1}\boldsymbol{q_{l}^{T}}\right\rangle _{l}\right)}{M_{1}'-M_{2}'}\\
 & =\frac{1}{LM}\left\langle \boldsymbol{q_{l}}\boldsymbol{W}^{-1}\boldsymbol{W}^{\beta}\boldsymbol{W}^{-1}\boldsymbol{q_{l}^{T}}\right\rangle _{l}=E\left[\mathcal{E}_{pred}\left(\hat{\Gamma};L'\right)\right]=\mathcal{E}_{pred}\left(\hat{\Gamma}\right).
\end{align*}
On the other hand, the form (\ref{eq:var_pred}) for $Var\left[\mathcal{E}_{pred}\left(\hat{\Gamma}\right)\right]$
is easily derived because $\mathcal{E}_{fit}\left(M_{1}';L_{1}'\right)$
and $\mathcal{E}_{fit}\left(M_{2}';L_{2}'\right)$ are independent
of each other. Subsequently, by Proposition \ref{prop:variance_of_error}, because $Var\left[\mathcal{E}_{fit}\left(M_{1}';L_{1}'\right)\right]$
and $Var\left[\mathcal{E}_{fit}\left(M_{2}';L_{2}'\right)\right]$
converge to $0$ as $L_{1}$ and $L_{2}$ go to infinity, $\hat{\mathcal{E}}_{pred}\left(\hat{\Gamma}\right)$ is consistent.
\end{proof}
When considering the theorem above, $\left(M_{1}',M_{2}'\right)$ and
$\left(L_{1}',L_{2}'\right)$ should be set as $M_{1}'\gg M_{2}'$,
$L_{1}'\gg1$, and $L_{2}'\gg1$ to precisely estimate $\mathcal{E}_{pred}\left(\hat{\Gamma}\right)$. 

\section{Data generation for network learning under the SABR model}

\subsection{the SABR model}

\begin{figure}
\centering{}\subfloat[S\&P500]{\centering{}\includegraphics[scale=0.7]{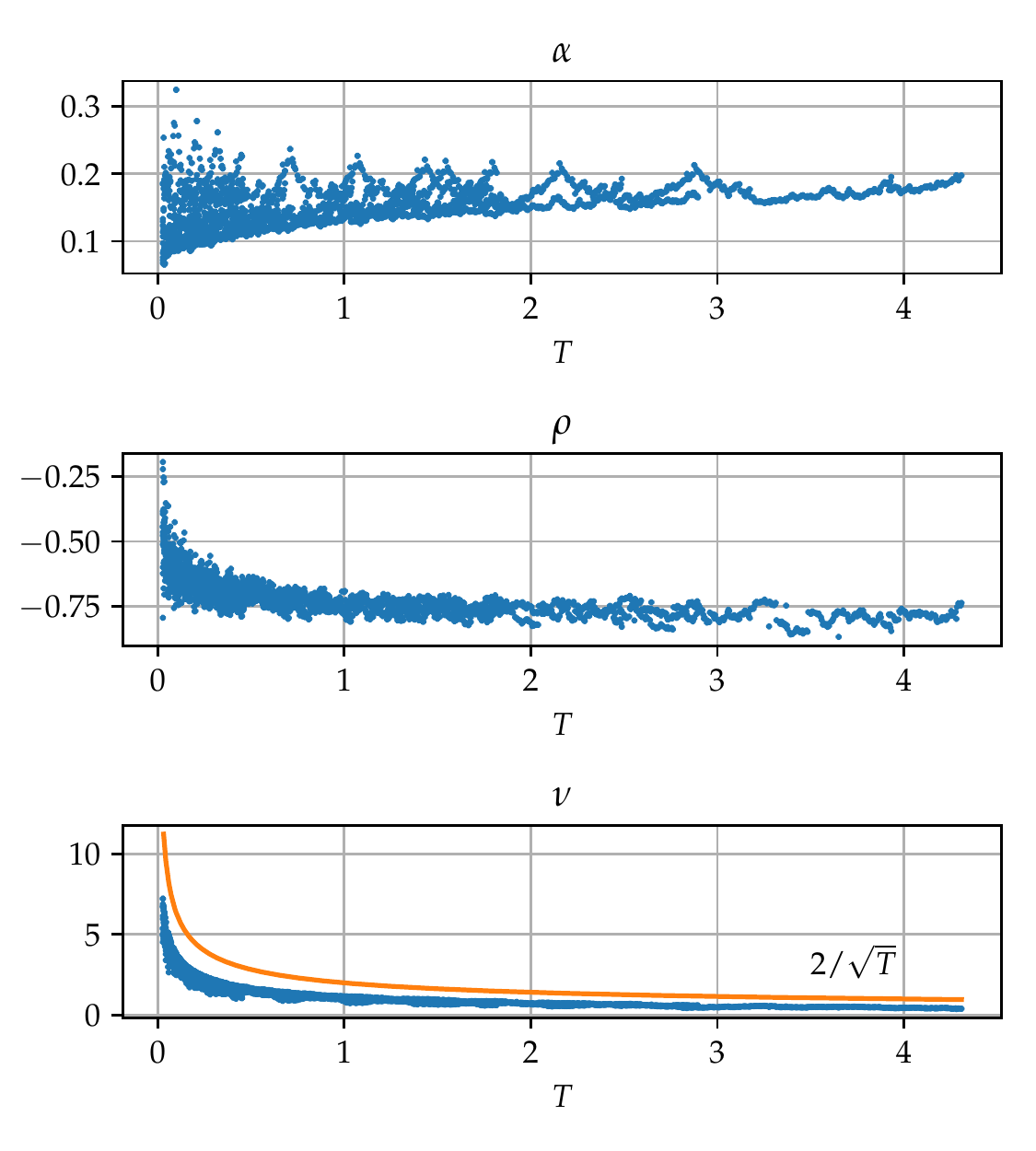}}
\subfloat[KOSPI]{\centering{}\includegraphics[scale=0.7]{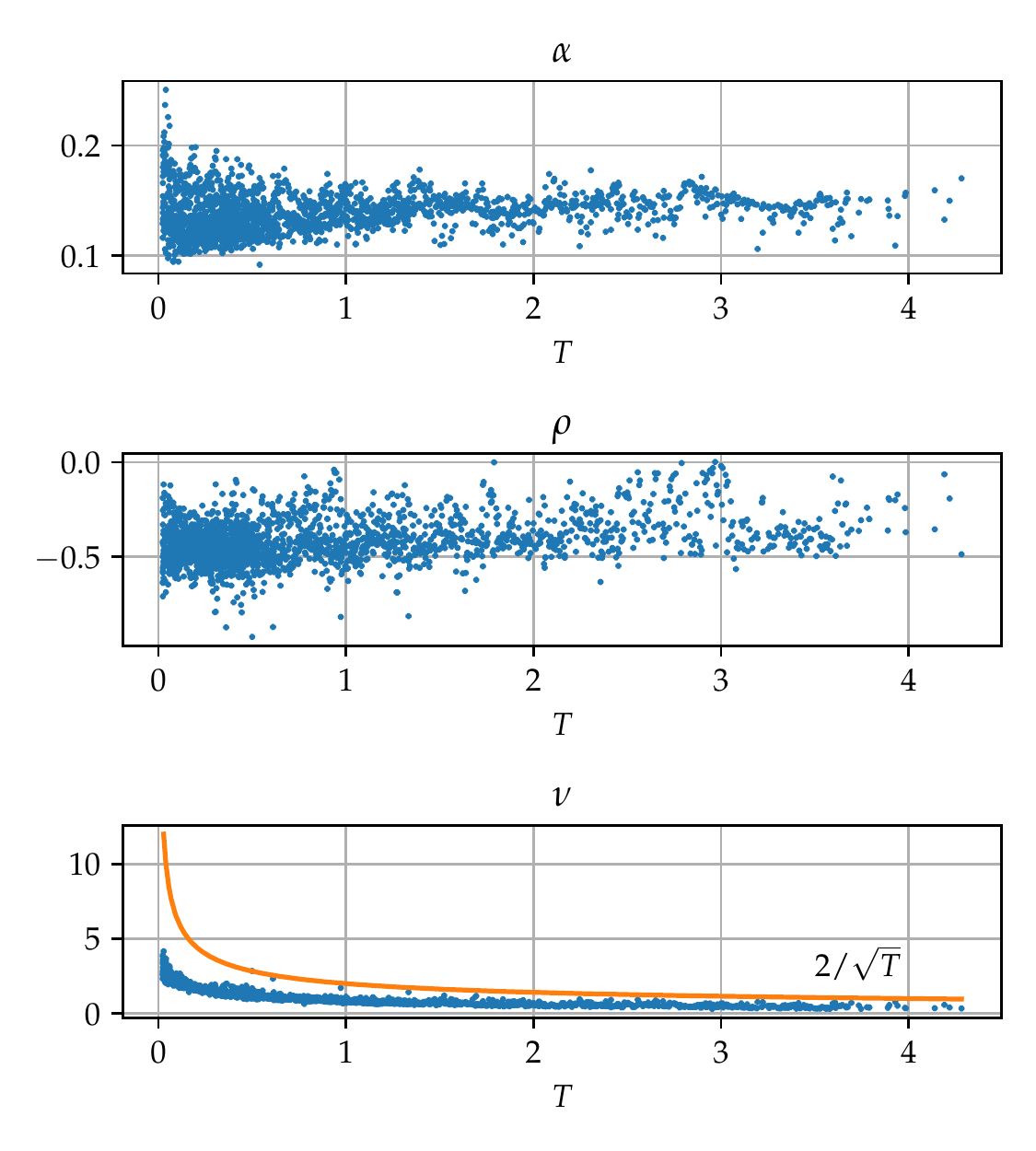}}\caption{\label{fig:raw_param}This figure shows the estimates of $\alpha_{0}(T)$, $\nu(T)$, and $\rho(T)$ for the SABR model where $\beta(T)=1$. The values are obtained using the option data for the S\&P 500 (left) and the KOSPI 200 (right) from April 2018 to March 2019.}
\end{figure}

The SABR model \citep{hagan2002managing} is expressed as the following stochastic differential equation (SDE):
\begin{gather*}
df_{t}=\alpha_{t}f_{t}^{\beta}dW_{t},\\
d\alpha_{t}=\nu\alpha_{t}dZ_{t},
\end{gather*}
where $f_{t}$ is the forward price of an underlying asset (i.e.,
stock and interest rate) at time $t$, and $W_{t}$ and $Z_{t}$ are
Brownian motions correlated with $\rho\in\left(-1,1\right)$. The
hidden state $\alpha_{t}$ and the parameters $\beta$, $\rho$, $\nu$
of the model have their respective roles in determining the shapes
of implied volatility surface $\sigma^{I}\left(T,K\right)$ (see \citet{rebonato2011sabr} for a more detailed explanation). The state $\alpha_{t}$ forms the backbone of the surface because the change of $\alpha_{t}$ causes a parallel shift upward of the surface. The volatility of volatility parameter $\nu$ handles the wings of the volatility surface
because it controls the curvature of the surface. Conversely, the elasticity $\beta$ and the correlation $\rho$ play similar roles in adjusting the slopes of skews on the surface. Thus, $\beta$ is commonly fixed as a constant from $0$ to $1$ to reduce model complexity. Aesthetic considerations result in $\beta=0$, $\beta=1/2$, and $\beta=1$, which are called the normal SABR, CIR (named after Cox, Ingersoll, and Ross) SABR, and the log-normal SABR, respectively. It is known that such
arbitrary choices of $\beta$ hardly ever decrease the fitting performance
of the SABR model \citep{west2005calibration,rebonato2011sabr}.
Likewise, \citet{bartlett2006hedging} developed a hedging method
less sensitive to particular values of $\beta$. From these studies,
we choose the log-normal SABR ($\beta=1$) so that $\alpha_{t}$ can
be regarded as the volatility of $f_{t}$.

The SABR model is mostly utilized as a fitting model to market volatilities
by maturity, for which the state $\alpha_{0}$ and parameters $\beta$,
$\nu$, and $\rho$ are usually parameterized as $\alpha_{0}(T)$, $\beta(T)$,
$\nu(T)$, and $\rho(T)$. Figure \ref{fig:raw_param} displays the estimates
of $\alpha_{0}(T)$, $\nu(T)$, and $\rho(T)$ when $\beta(T)=1$.
These are derived utilizing the option data for the S\&P 500 (left) and
the KOSPI 200 (right) from April 2018 to March 2019. The calibration
is performed by Korean Asset Pricing, a bond rating agency located
in Korea. From the figure, one can observe that all $\alpha_{0}$
and $\rho$ belong to $(0.01,0.5)$ and $(-0.99,0.1)$, respectively,
and all of $\nu$ are lower than the baseline $2/\sqrt{T}$. Moreover,
it seems that the values tend to become more unstable as $T$ gets
shorter, particularly for $\nu$. This association may result because the SABR model ignores short-term events such as fast-mean-reverting volatility. 

Let us assume a fair price $c$ for a vanilla option under the SABR
model. Under the risk-neutral pricing framework \citep{shreve2004stochastic},
we can induce the pricing formula by solving the integral 
\begin{align*}
c\left(t,K\right) & =\int_{0}^{\infty}q\left(f_{T};K\right)p_{t}\left(f_{T}\right)df_{T},
\end{align*}
or the partial differential equation (PDE)
\begin{gather*}
\frac{\partial c}{\partial t}=\frac{1}{2}\alpha^{2}f^{2\beta}\frac{\partial^{2}c}{\partial f^{2}}+\frac{1}{2}\nu^{2}\alpha^{2}\frac{\partial^{2}c}{\partial\alpha^{2}}+\rho\nu\alpha^{2}f^{\beta}\frac{\partial^{2}c}{\partial f\partial\alpha},\\
c\left(T,K\right)=q\left(f_{T};K\right),
\end{gather*}
where $p_{t}$ is the density of $f_{T}$, $K$ and $T$ are the strike
and maturity of the option, respectively, and $q$ $\left(\cdot;K\right)$
is the payoff function of the option with $K$. Regrettably, any exact
pricing formulas cannot be derived for the option because the integral
and the PDE are fairly hard to solve analytically. Instead, it
is possible to derive a wide range of approximate formulas for the
implied volatilities of the options \citep{hagan2002managing,obloj2008fine,henry2008analysis,paulot2009asymptotic,wu2012series,antonov2013sabr}.
In fact, the aforementioned notion explains why the SABR model is so popular in practice. 

\begin{figure}
\centering{}\includegraphics[scale=0.75]{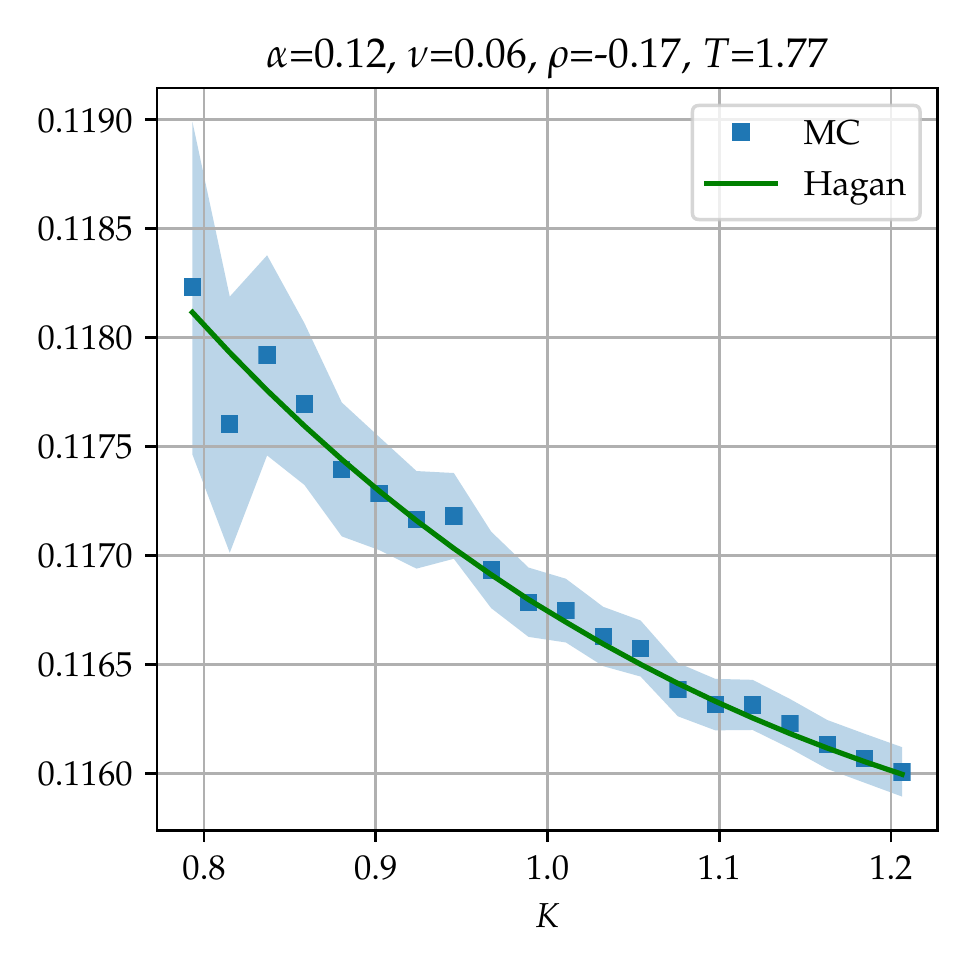}$\qquad$\includegraphics[scale=0.75]{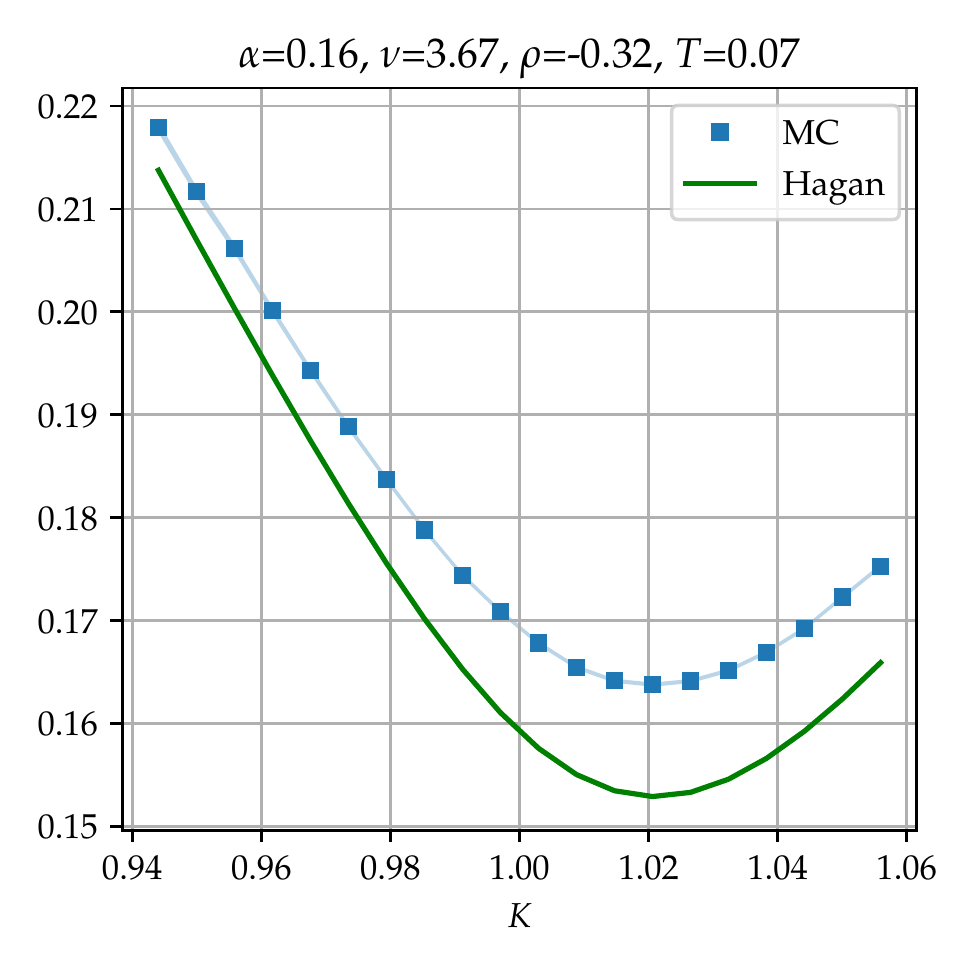}\caption{\label{fig:Hagan}These figures are drawn to compare the implied volatilities provided by the Hagan formula (\ref{eq:Hagan}) with the values of the MC. For the MC, we simulate 10 million paths for the time interval 0.002.
The blue region indicates the 99\% confidence interval for the MC.}
\end{figure}

Additionally, only for the case $\beta=1$, we briefly mention
the most well-known asymptotic formulas for the implied volatility
$\sigma^{I}$, which was found by \citet{hagan2002managing} as follows:
\begin{align}
\sigma^{I}\left(T,K\right) & \approx\alpha_{0}\frac{z}{\chi\left(z\right)}\left\{ 1+\left[\frac{1}{4}\rho\alpha_{0}\nu+\frac{1}{24}\left(2-3\rho^{2}\right)\nu^{2}\right]T\right\} ,\label{eq:Hagan}
\end{align}
where 
\begin{gather*}
z=\frac{\nu}{\alpha_{0}}\log\frac{f_{0}}{K},\quad\chi\left(z\right)=\log\left\{ \frac{\sqrt{1-2\rho z+z^{2}}+z-\rho}{1-\rho}\right\} .
\end{gather*}
As the above equations are derived using an asymptotic technique,
it can be applied only for the option with a short maturity $T$ and
a strike $K$ close to $f_{0}$. If one of the assumptions does not
hold, this formula yields inexact prices, but even if all of the assumptions are satisfied, it does not always provide consistent accuracy. Figure \ref{fig:Hagan} compares the implied volatilities formula
with the MC-based values. Its accuracy is similar to that of the MC even when $T$ is large (left), and vice versa (right).
For the MC simulation, we simulate 10 million paths for the time
interval 0.002. The blue region indicates the 99\% confidence interval
for the MC. Nevertheless, formula (\ref{eq:Hagan}) is still popular
due to its simplicity, especially in global over-the-counter interest
rate derivatives market.

\subsection{Data generation for network learning}

In this subsection, we explain the method to generate extensive approximate
implied volatilities $\sigma_{approx,l}^{I}$ for exact implied volatilities
$\sigma_{true,l}^{I}$ $\left(l=1,\cdots,L\right)$ under the SABR
model where $\beta=1$, which will be used to train and test neural
networks later. The volatilities $\sigma_{approx,l}^{I}$ are grouped
into $mn$ data to form $\tilde{L}$ surfaces $\sigma_{approx,s}^{I}$$(T_{s,k_{1}},K_{s,k_{2}})$, 
with respect to $T_{s,k_{1}}$ and $K_{s,k_{2}}$, where $\tilde{L}=L/(mn)$,
$s=1,\cdots,\tilde{L}$, $k_{1}=1,\cdots,m$, and $k_{2}=1,\cdots,n$.
In other words, $L$ is the total number of generated data, and
$\tilde{L}$ is the number of volatility surfaces (a sort of partition
of the data) with $mn$ grid points.

When constructing grid points for the surfaces, $K_{s,k_{2}}$ is
set as a function of $T_{s,k_{1}}$; that is, $K_{s,k_{2}}=K_{s,k_{2}}(T_{s,k_{1}})$.
This mechanism is intended to widen the width of strike range $[K_{s,1},K_{s,n}]$
as $T_{s,k_{1}}$ gets larger, in other words,
\[
\log\frac{K_{s,n}}{f_{0,s}}-\log\frac{K_{s,1}}{f_{0,s}}\propto{\rm std}\left[\log\frac{f_{T_{s,k_{1}},s}}{f_{0,s}}\right],
\]
where $f_{t,s}$ is the forward price at $t$ under the parameters
$\alpha_{0,s}$, $\tilde{\nu}_{s}$$\left(:=\nu_{s}/\sqrt{T_{s,k_{1}}}\right)$,
and $\rho_{s}$. Note that the parameter $\tilde{\nu}_{s}$ depends
on $T_{s,k_{1}}$. As mentioned before, practitioners usually fit the SABR
model to the market data $\sigma^{I}(T)$ for each maturity $T$ separately.
If so, as shown in Figure 1, the estimates of $\nu$ tend to
get larger as $T$ becomes shorter in market data. Hence, $\tilde{\nu}_{s}$
is set up in a way to capture the phenomenon.

\begin{figure}
\centering{}\includegraphics[scale=0.62]{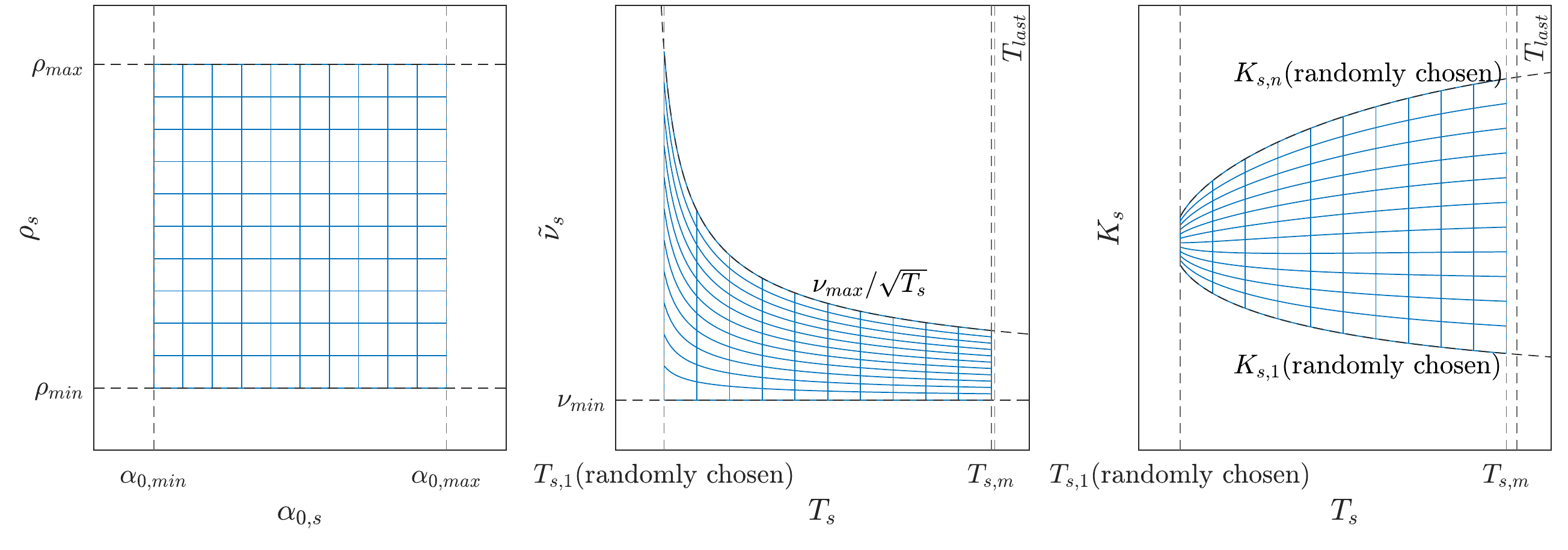}\caption{An illustration of the grid of $\alpha_{0,s}$, $\rho_{s}$, $\tilde{\nu}_{s}$,
$K_{s}$, and $T_{s}$.}
\end{figure}
Let us explain the construction process of the grid points in more
detail. Suppose that data is generated until the time $T_{last}$.
The maturity $T_{s,k_{1}}$ and strike $K_{s,k_{2}}$ for the $s$th
surface $\sigma_{approx,s}^{I}(T_{s,k_{1}},K_{s,k_{2}})$ are randomly
chosen as follows:
\begin{enumerate}
\item Set a time grid interval $\Delta T=T_{last}/m$, initialize the first
time point $T_{s,1}$ randomly in $\left(0,\Delta T\right]$, and
choose the other points equidistantly by $T_{s,k_{1}}=T_{s,1}+\left(k_{1}-1\right)\Delta T$.
\item Determine the start point $K_{s,1}$ and the end point $K_{s,n}$
of $K_{s}$-range by the formula 
\begin{align*}
K_{s,1} & =f_{0,s}\exp\left(-\frac{1}{2}\frac{\alpha_{0,s}}{\tilde{\nu}_{s}^{2}}\left(\exp\left\{ \tilde{\nu}_{s}^{2}T_{s,k_{1}}\right\} -1\right)-\eta_{f}\frac{\alpha_{0,s}}{\tilde{\nu}_{s}}\left(\exp\left\{ \tilde{\nu}_{s}^{2}T_{s,k_{1}}\right\} -1\right)^{1/2}\right),\\
K_{s,n} & =f_{0,s}\exp\left(-\frac{1}{2}\frac{\alpha_{0,s}^{2}}{\tilde{\nu}_{s}^{2}}\left(\exp\left\{ \tilde{\nu}_{s}^{2}T_{s,k_{1}}\right\} -1\right)+\eta_{f}\frac{\alpha_{0,s}}{\tilde{\nu}_{s}}\left(\exp\left\{ \tilde{\nu}_{s}^{2}T_{s,k_{1}}\right\} -1\right)^{1/2}\right),
\end{align*}
where $\eta_{f}$ follows the uniform distribution ${\rm U}\left(0.842,2.576\right)$,
and the hidden state $\alpha_{0,s}$ and parameters $\nu_{s}$ (not
$\tilde{\nu}_{s}$) and $\rho_{s}$ are also sampled uniformly within
predetermined limits. That is, 
\[
\alpha_{0,s}\sim{\rm U}\left(\alpha_{0,min},\alpha_{0,max}\right),\;\nu_{s}\sim{\rm U}\left(\nu_{min},\nu_{max}\right),\;\rho_{s}\sim{\rm U}\left(\rho_{min},\rho_{max}\right).
\]
 
\item Equidistantly partition the interval $\left[K_{s,1},K_{s,n}\right]$
for each $T_{s,k_{1}}$ by $K_{s,k_{2}}=K_{s,1}+\left(k_{2}-1\right)\Delta K$,
where $\Delta K=\left(K_{s,n}-K_{s,1}\right)/n$.
\end{enumerate}
If the number $\tilde{L}$ of the surfaces is reasonably large, the
random numbers $\{(\alpha_{0,s},\nu_{s},\rho_{s}):s=1,\cdots,\tilde{L}\}$
may fill most of the parameter space $\left\{ \alpha_{0,min}\leq\alpha_{0}\leq\alpha_{0,max}\right\} \times\left\{ \nu_{min}\leq\nu\leq\nu_{max}\right\} \times\left\{ \rho_{min}\leq\rho\leq\rho_{max}\right\} $
evenly and densely. In addition, we suppose $f_{0,s}=1$ without loss
of generality. Note that $f_{t,s}$ of the SABR model where $\beta=1$
can be expressed as $d\left(f_{t,s}/f_{0,s}\right)=\alpha_{t,s}\left(f_{t,s}/f_{0,s}\right)dW_{t}$.
This notion implies that an option price $c$ is homogeneous of degree one in both $f_{0,s}$ and $K$ under the SABR model, that is, $c=c(K/f_{0,s})$ (\citet{garcia2000pricing}), which supports the assumption $f_{0,s}=1$.
Figure 3 is an illustration of the grids $\alpha_{0,s}$, $\rho_{s}$,
$\tilde{\nu}_{s}$, $K_{s}$, and $T_{s}$. The grid for $\alpha_{0}$
and $\rho$ is partitioned evenly, the partition $\left\{ T_{s,k_{1}}\right\} $
is chosen randomly, but $T_{s,1}$ cannot be less than $0$, and $T_{s,m}$
cannot be greater than $T_{last}$. The upper boundary of $\tilde{\nu}_{s}$
is dependent on $T_{s,k_{1}}$ because $\tilde{\nu}_{s}=\nu_{max}/\sqrt{T_{s,k_{1}}}$,
where $\nu_{max}$ is selected based on Figure 1. The strike boundary
also relies on $T_{s,k_{1}}$, and it widens and narrows under
the influence of the value of the random variable $\eta_{f}$. All
grid points inside the boundaries are split equidistantly.

To obtain the approximate implied volatility $\sigma_{approx,s}^{I}(T_{s,k_{1}},K_{s,k_{2}})$,
we simulate $N$ paths of the SABR model ($\beta=1$) under the parameters
$\left(\alpha_{0,s},\tilde{\nu}_{s},\rho_{s}\right)$ using the Monte-Carlo
Euler scheme, which can be expressed as follows:
\begin{align*}
f_{t+\Delta t,s} & =f_{t,s}+\alpha_{t,s}f_{t,s}\left(e_{t,s}\sqrt{\Delta t}\right), \\
\alpha_{t+\Delta t,s} & =\alpha_{t,s}+\tilde{\nu}_{s}\alpha_{t,s}\left(\rho_{s}e_{t,s}\sqrt{\Delta t}+\sqrt{1-\rho_{s}^{2}}\tilde{e}_{t,s}\sqrt{\Delta t}\right),
\end{align*}
where $t=\Delta t,2\Delta t,\cdots,T_{s,m}-\Delta t$, $f_{0,s}=1$, and
$\left(e_{t,s},\tilde{e}_{t,s}\right)$ are independent standard normal
random variables. Further, the MC gives an approximate price $c_{approx,s,k_{1},k_{2}}$
for the true option price $c_{true,s,k_{1},k_{2}}$ under the SABR
model in the following way:
\[
c_{approx,s,k_{1},k_{2}}=\frac{1}{N}\sum_{j=1}^{N}c_{approx,s,k_{1},k_{2},j}=\frac{1}{N}\sum_{j=1}^{N}q\left(f_{T_{l,k_{1}},s,j}, K_{s,k_{2}}\right), 
\]
where the subscript $j$ represents that the value originated from
the $j$th path, and $q\left(\cdot;K_{s,k_{2}}\right)$ is the option
payoff for strike $K_{s,k_{2}}$. Price $c_{approx,s,k_{1},k_{2}}$ 
is then converted to its implied volatility $\sigma_{approx,s,k_{1},k_{2}}^{I}$.
We use Powell's method to obtain the implied volatility by minimizing
$({\rm BS}(\sigma_{approx,s,k_{1},k_{2}}^{I})-c_{approx,s,k_{1},k_{2}})^{2}$,
where ${\rm BS}\left(\sigma\right)$ is the Black-Scholes formula
for the option. In this process, we only utilize out-of-the-money
(OTM) call and put options. The price of OTM call options explodes occasionally. Therefore, if the standard deviation of the simulations exceeds by 100 times the average, the data were excluded from the experiment (approximately 3.6\% of the data are excluded in this manner). 

\begin{table}
\centering{}%
\begin{tabular}{cccc}
\toprule 
\multirow{2}{*}{dataset} & $N$ & $\tilde{L}$ & $L$\tabularnewline
 & (\# of paths) & (\# of surface) & (\# of total data)\tabularnewline
\midrule
\midrule 
training set ($\mathcal{D}_{train}$) & 500k & 1.2M & 480M\tabularnewline
validation set ($\mathcal{D}_{validate}$) & 500k & 0.1M & 40M\tabularnewline
test set ($\mathcal{D}_{test}$) & 500k & 0.1M & 40M\tabularnewline
more accurate set ($\mathcal{D}_{test}$') & 12.5M & 0.25M & 100M\tabularnewline
\midrule 
total &  & 1.65M & 660M\tabularnewline
\bottomrule
\end{tabular}\caption{\label{tab:datasets}This table displays the datasets for our experiments,
which contain approximate implied volatilities $\sigma_{approx,l}^{I}$
for $\sigma_{true,l}^{I}$ ($l=1,\cdots,L$). They are generated by
the MC with $N$ paths for a time interval of 0.002. Furthermore, $\sigma_{approx,l}^{I}$
are grouped to form $\tilde{L}$ surfaces. The symbols "k" and "M"
indicate one thousand and one million, respectively.}
\end{table}
The hyperparameters are chosen as follows: $m=20$, $n=20$, $\alpha_{0,min}=0.01$,
$\alpha_{0,max}=0.5$, $\nu_{min}=0.01$, $\nu_{max}=2$, $\rho_{min}=-0.99$,
$\rho_{max}=0.1$, and $T_{last}=2$. With the hyperparameters, we
separately make 1.2M surfaces ($N=5{\rm E}+5$, $L=4.8{\rm E}+8$)
for training, 0.1M surfaces ($N=5{\rm E}+5$, $L=4{\rm E}+7$) for
a validation, 0.1M surfaces for a test with simulation accuracy $\Delta t=0.002$
and $N=5{\rm E}+5$. In fact, we generate 250k paths with the antithetic
variate method in the experiment, but it is known that this approach
is superior to when 500k paths are generated without the method.
Moreover, we create 0.25M surfaces with a higher accuracy
$N=1.25{\rm E}+7$ ($L=1{\rm E}+8$) while keeping $\Delta t$ as
$0.002$. These additional data are employed for error analysis,
which will be explained in the following section. For convenience,
the four kinds of datasets are summarized in Table \ref{tab:datasets}.
Contrarily, to generate the datasets, we performed the MC in
parallel using many GPUs (GeForce GTX 1080 TI $\times8$, GeForce
RTX 2080 TI $\times16$, Tesla V100 $\times2$). Despite using many
GPUs, the procedure took about a month to complete.

\section{\label{sec:test}Network-based prediction of implied volatilities }

In this section, we train neural networks of various structures using
the training dataset $\mathcal{D}_{train}$ and the validation dataset
$\mathcal{D}_{validate}$. Subsequently, we predict the data in the test dataset
$\mathcal{D}_{test}$ with the best-performing network among them.
We also predict the data in the more accurate test dataset $\mathcal{D}_{test}'$ to evaluate the networks by estimating $\mathcal{E}_{pred}\left(\hat{\Gamma}\right)$
with Theorem \ref{thm:main_theorem} in Section \ref{sec:theorem}.
In addition, new notations are introduced confirming that $\mathcal{E}_{fit}\left(\hat{\Gamma},\mathcal{D}\right)=\mathcal{E}_{fit}\left(\hat{\Gamma},M';L'\right)$,
$\mathcal{E}_{pred}\left(\hat{\Gamma};\mathcal{D}\right)=\mathcal{E}_{pred}\left(\hat{\Gamma};L'\right)$,
$\mathcal{E}_{approx}\left(\mathcal{D}\right)=\mathcal{E}_{approx}\left(M';L'\right)$,
$\mathcal{N}_{pred}\left(\hat{\Gamma},\mathcal{D}\right)=\mathcal{N}_{pred}\left(\hat{\Gamma},M';L'\right)$
for the dataset $\mathcal{D}$ with data length $L'$ for $M'$ simulations
(i.e., see below for the definition of $\mathcal{N}_{pred}$). The notations
are used when it is desirable to emphasize $\mathcal{D}$ more than
$L'$ and $M'$ in the context.

The expected number $\mathcal{N}_{pred}$ of virtual simulations to
achieve $\mathcal{E}_{pred}\left(\hat{\Gamma}\right)$ is defined
as follows:
\begin{equation}
\mathcal{N}_{pred}\left(\hat{\Gamma},M';L'\right)=M'\frac{\mathcal{E}_{fit}\left(\hat{\Gamma},M';L'\right)-\mathcal{E}_{pred}\left(\hat{\Gamma}\right)}{\mathcal{E}_{pred}\left(\hat{\Gamma}\right)},\label{eq:N_pred}
\end{equation}
which is utilized as an indicator of the network performance, along
with $\mathcal{E}_{pred}\left(\hat{\Gamma}\right)$. The definition
is plausible because $\mathcal{E}_{approx}\left(M';L'\right)=\mathcal{E}_{fit}\left(\hat{\Gamma},M';L'\right)-\mathcal{E}_{pred}\left(\hat{\Gamma};L'\right)$,
and it is expected that $\left|\mathcal{E}_{pred}\left(\hat{\Gamma}\right)-\mathcal{E}_{pred}\left(\hat{\Gamma};L'\right)\right|$
is considerably small. For instance, suppose that $\mathcal{N}_{pred}$ is
approximately one million. This supposition implies that simulations should be performed one million times so that the MC can attain the accuracy of the network.

We use an extensive feedforward neural network with millions of weights.
By increasing network size when possible, we expect that approximation
capability will be maximized. It accepts the following five inputs
\begin{gather*}
\boldsymbol{T}_{s}=\left[\begin{array}{cccc}
T_{s,1} & T_{s,1} & \cdots & T_{s,1}\\
T_{s,2} & T_{s,2} & \cdots & T_{s,2}\\
\vdots & \vdots & \ddots & \vdots\\
T_{s,m} & T_{s,m} & \cdots & T_{s,m}, 
\end{array}\right],\quad\boldsymbol{K}_{s}=\left[\begin{array}{cccc}
K_{s,1,1} & K_{s,1,2} & \cdots & K_{s,1,n}\\
K_{s,2,1} & K_{s,2,2} & \cdots & K_{s,2,n}\\
\vdots & \vdots & \ddots & \vdots\\
K_{s,m,1} & K_{s,m,2} & \cdots & K_{s,m,n}
\end{array}\right],\\
\boldsymbol{\alpha}_{0,s}=\left[\begin{array}{cccc}
\alpha_{0,s} & \alpha_{0,s} & \cdots & \alpha_{0,s}\\
\alpha_{0,s} & \alpha_{0,s} & \cdots & \alpha_{0,s}\\
\vdots & \vdots & \ddots & \vdots\\
\alpha_{0,s} & \alpha_{0,s} & \cdots & \alpha_{0,s}
\end{array}\right],\quad\boldsymbol{\nu}_{s}=\left[\begin{array}{cccc}
\nu_{s} & \nu_{s} & \cdots & \nu_{s}\\
\nu_{s} & \nu_{s} & \cdots & \nu_{s}\\
\vdots & \vdots & \ddots & \vdots\\
\nu_{s} & \nu_{s} & \cdots & \nu_{s}
\end{array}\right],\quad\boldsymbol{\rho}_{s}=\left[\begin{array}{cccc}
\rho_{s} & \rho_{s} & \cdots & \rho_{s}\\
\rho_{s} & \rho_{s} & \cdots & \rho_{s}\\
\vdots & \vdots & \ddots & \vdots\\
\rho_{s} & \rho_{s} & \cdots & \rho_{s}
\end{array}\right],
\end{gather*}
where $K_{s,k_{1},k_{2}}=K_{s,k_{2}}\left(T_{s,k_{1}}\right)$, and produces the following SABR volatilities:
\[
\boldsymbol{\sigma}_{net,s}^{I}=\left[\begin{array}{ccc}
\sigma_{net,s}^{I}\left(T_{s,1},K_{s,1,1},\alpha_{0,s},\nu_{s},\rho_{s}\right) & \cdots & \sigma_{net,s}^{I}\left(T_{s,1},K_{s,1,n},\alpha_{0,s},\nu_{s},\rho_{s}\right)\\
\vdots & \ddots & \vdots\\
\sigma_{net,s}^{I}\left(T_{s,m},K_{s,m,1},\alpha_{0,s},\nu_{s},\rho_{s}\right) & \cdots & \sigma_{net,s}^{I}\left(T_{s,1},K_{s,m,n},\alpha_{0,s},\nu_{s},\rho_{s}\right)
\end{array}\right].
\]
It is worthy to recall that $\boldsymbol{\sigma}_{approx,s}^{I}$
($s=1,2,\cdots,\tilde{L}$) is also generated as a two-dimensional form
via
\[
\boldsymbol{\sigma}_{approx,s}^{I}=\left[\begin{array}{ccc}
\sigma_{appox,s}^{I}\left(T_{s,1},K_{s,1,1},\alpha_{0,s},\nu_{s},\rho_{s}\right) & \cdots & \sigma_{approx,s}^{I}\left(T_{s,1},K_{s,1,n},\alpha_{0,s},\nu_{s},\rho_{s}\right)\\
\vdots & \ddots & \vdots\\
\sigma_{approx,s}^{I}\left(T_{s,m},K_{s,m,1},\alpha_{0,s},\nu_{s},\rho_{s}\right) & \cdots & \sigma_{approx,s}^{I}\left(T_{s,1},K_{s,m,n},\alpha_{0,s},\nu_{s},\rho_{s}\right)
\end{array}\right].
\]
This kind of approach would help a network in filtering out simulation
noises by checking adjacent values. \citet{dimitroff2018volatility}
and \citet{bayer2019deep} adopted similar approaches. 

The optimal weights $\hat{\Gamma}$ are determined by minimizing the sum
of squared differences between $\boldsymbol{\sigma}_{net,s}^{I}$
with $\boldsymbol{\sigma}_{approx,s}^{I}$ in $\mathcal{D}_{train}$.
To this end, the adaptive moment estimation (ADAM, \citet{kingma2014adam})
is used with the batch size 100. The learning rate is initially set
to be $1e^{-5}$, but it is reduced by a factor of 10 every time the
loss value for $\mathcal{D}_{validate}$ is not improved. When the
learning rate reaches the value $1e^{-8}$, the training is finished.
We refer to \citet{liu2019pricing} for the configuration. 

\begin{table}
\centering{}%
\begin{tabular}{cccccccccc}
\toprule 
\multirow{3}{*}{\# of nodes } & \multicolumn{4}{c}{$\mathcal{E}_{pred}\left(\hat{\Gamma}\right)$} &  & \multicolumn{4}{c}{$\mathcal{N}_{pred}\left(\hat{\Gamma},\mathcal{D}_{test}\right)$}\tabularnewline
\cmidrule{2-5} \cmidrule{3-5} \cmidrule{4-5} \cmidrule{5-5} \cmidrule{7-10} \cmidrule{8-10} \cmidrule{9-10} \cmidrule{10-10} 
 & \multicolumn{4}{c}{\# of layers} &  & \multicolumn{4}{c}{\# of layers}\tabularnewline
 & 4 & 5 & 6 & 7 &  & 4 & 5 & 6 & 7\tabularnewline
\midrule
3,000 & 2.22E-7 & 2.13E-7 & 2.27E-7 & 2.04E-7 &  & 11.80M & 12.31M & 11.52M & 12.84M\tabularnewline
5,000 & 2.23E-7 & 2.05E-7 & 2.05E-7 & 2.03E-7 &  & 11.73M & 12.81M & 12.81M & 12.88M\tabularnewline
7,000 & $\boldsymbol{2.03E}$-$\boldsymbol{7}$ & 2.10E-7 & 2.06E-7 & 2.06E-7 &  & $\boldsymbol{12.95M}$ & 12.50M & 12.71M & 12.72M\tabularnewline
\bottomrule
\end{tabular}\caption{\label{tab:structure}This table denotes the values of the two indicators of network performance, $\mathcal{E}_{pred}\left(\hat{\Gamma}\right)$
and $\mathcal{N}_{pred}\left(\hat{\Gamma},\mathcal{D}_{test}\right)$,
for various network structures ("M" indicates one million). }
\end{table}

We train and test various structures of the network while varying the
numbers of layers and nodes per each hidden layer. Table \ref{tab:structure}
reveals $\mathcal{E}_{pred}\left(\hat{\Gamma}\right)$ and $\mathcal{N}_{pred}\left(\hat{\Gamma},\mathcal{D}_{test}\right)$
for the networks, which are induced using Theorem \ref{thm:main_theorem}
with $\mathcal{E}_{fit}\left(\hat{\Gamma},\mathcal{D}_{test}\right)$
and $\mathcal{E}_{fit}\left(\hat{\Gamma},\mathcal{D}_{test}'\right)$.
In addition, $\mathcal{E}_{pred}\left(\hat{\Gamma}\right)$ are spread from $2.03\times10^{-7}$
to $2.22\times10^{-7}$, and $\mathcal{N}_{pred}\left(\hat{\Gamma},\mathcal{D}_{test}\right)$
are distributed from 11.52M to 12.95M ("M" indicates one million).
After observing the figures on the table, we conclude that there are
no significant differences between the performances of the networks.
However, notably, the values tend to improve slightly
as the numbers of layers and nodes increase. We chose the network
with 4 layers and 7,000 nodes as the best performance network. Thus, we
will only continue subsequent tests for the chosen network.

\begin{figure}
\begin{centering}
\includegraphics[scale=0.91]{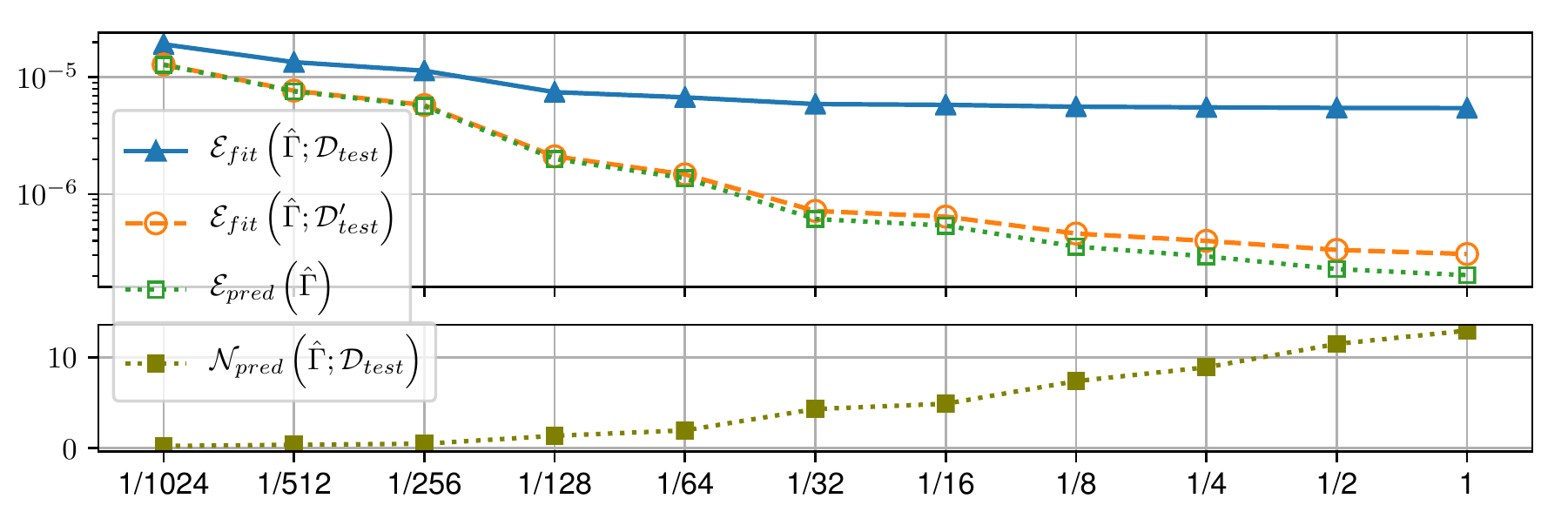}
\par\end{centering}
\caption{\label{fig:data_size}This figure draws $\mathcal{E}_{fit}\left(\hat{\Gamma},\mathcal{D}_{test}\right)$,
$\mathcal{E}_{fit}\left(\hat{\Gamma},\mathcal{D}_{test}'\right)$,
$\mathcal{E}_{pred}\left(\hat{\Gamma}\right)$, and $\mathcal{N}_{pred}\left(\hat{\Gamma},\mathcal{D}_{test}\right)$
with respect to training data size.}
\end{figure}

Subsequently, we analyze the effect of training data size $L$ on
network performance. Figure \ref{fig:data_size} draws $\mathcal{E}_{fit}\left(\hat{\Gamma},\mathcal{D}_{test}\right)$,
$\mathcal{E}_{fit}\left(\hat{\Gamma},\mathcal{D}_{test}'\right)$,
$\mathcal{E}_{pred}\left(\hat{\Gamma}\right)$, and $\mathcal{N}_{pred}\left(\hat{\Gamma},\mathcal{D}_{test}\right)$
with respect to $L$. To obtain the values, training and test are
repeated using part of the training set $\mathcal{D}_{train}$. For
instance, the values for the "$1/128$" subset originate from
the network trained by only using $1/128$ part of $\mathcal{D}_{train}$.
A subset contains other subsets with smaller sizes. For example,
the subset "$1/256$" contains the subset "$1/512$", and the subset
"$1/512$" contains the subset "$1/1024$". In the figure, $\mathcal{E}_{fit}\left(\hat{\Gamma},\mathcal{D}_{test}\right)$,
$\mathcal{E}_{fit}\left(\hat{\Gamma},\mathcal{D}_{test}'\right)$,
and $\mathcal{E}_{pred}\left(\hat{\Gamma}\right)$ decrease, and $\mathcal{N}_{pred}\left(\hat{\Gamma}\right)$
increases as the training data size $L$ increases. This phenomenon
can be understood through Proposition \ref{prop:expectation_of_error}
and Theorem \ref{thm:main_theorem}. Accordingly, it is established
that
\[
E\left[\mathcal{E}_{fit}\left(\hat{\Gamma},\mathcal{D}\right)\right]=E\left[\mathcal{E}_{pred}\left(\hat{\Gamma};\mathcal{D}\right)\right]+E\left[\mathcal{E}_{approx}\left(\mathcal{D}\right)\right]=\frac{1}{LM}\left\langle \boldsymbol{q_{l}}\boldsymbol{W}^{-1}\boldsymbol{W}^{\beta}\boldsymbol{W}^{-1}\boldsymbol{q_{l}^{T}}\right\rangle _{l,L'}+\frac{1}{M'}\left\langle \beta_{l}^{2}\right\rangle _{l,L'},
\]
\[
E\left[\mathcal{E}_{pred}\left(\hat{\Gamma}\right)\right]=\frac{1}{LM}\left\langle \boldsymbol{q_{l}}\boldsymbol{W}^{-1}\boldsymbol{W}^{\beta}\boldsymbol{W}^{-1}\boldsymbol{q_{l}^{T}}\right\rangle _{l},\quad E\left[\mathcal{N}_{pred}\left(\hat{\Gamma};D\right)\right]=\frac{LM}{M'}\left\langle \beta_{l}^{2}\right\rangle _{l,L'}\left\langle \boldsymbol{q_{l}}\boldsymbol{W}^{-1}\boldsymbol{W}^{\beta}\boldsymbol{W}^{-1}\boldsymbol{q_{l}^{T}}\right\rangle ^{-1}
\]
for the dataset $\mathcal{D}$ with data length $L'$ for $M'$ simulations
($M$ is the number of simulations to generate the data in $\mathcal{D}_{train}$).
By the formulas above, $\mathcal{E}_{pred}\left(\hat{\Gamma};\mathcal{D}\right)$
is reduced with a high probability as $L$ increases, but
$\mathcal{E}_{approx}\left(\mathcal{D}\right)$ is independent of
$L$. Thus, as the training data size $L$ grows, $\mathcal{E}_{fit}\left(\hat{\Gamma},\mathcal{D}\right)$
probably decreases and converges to $\mathcal{E}_{app}\left(\mathcal{D}\right)$.
It also seems reasonable that $\mathcal{E}_{pred}\left(\hat{\Gamma}\right)$
and $\mathcal{N}_{pred}\left(\hat{\Gamma};D\right)$ monotonically
decrease and increase, respectively. Nonetheless, according to the formulas
above, $\mathcal{E}_{pred}\left(\hat{\Gamma}\right)$ and $\mathcal{N}_{pred}\left(\hat{\Gamma};D\right)$
should approximately be half and doubled, respectively. Nevertheless,
in the graphs, the decrease rate of $\mathcal{E}_{pred}\left(\hat{\Gamma}\right)$
is higher than 0.5, and the increase rate of $\mathcal{N}_{pred}\left(\hat{\Gamma};D\right)$
is lower than 2. These results may be because the optimal weights $\hat{\Gamma}$
are all different depending on the training data size $L$, highlighting
that matrices such as $\boldsymbol{W}$ and $\boldsymbol{W}^{\beta}$
adjust if $L$ changes. Further, we guess that the approximation of the Hessian for the loss function in (\ref{eq:jaco_hess}) leads to these outcomes. We leave this issue for future endeavors.
\begin{table}[t]
\begin{centering}
\begin{tabular}{ccccccccccccccc}
\toprule 
\multirow{2}{*}{} & \multicolumn{2}{c}{$\alpha_{0}$} &  & \multicolumn{2}{c}{$\nu$} &  & \multicolumn{2}{c}{$\rho$} &  & \multicolumn{2}{c}{$K$} &  & \multicolumn{2}{c}{$T$}\tabularnewline
\cmidrule{2-3} \cmidrule{3-3} \cmidrule{5-6} \cmidrule{6-6} \cmidrule{8-9} \cmidrule{9-9} \cmidrule{11-12} \cmidrule{12-12} \cmidrule{14-15} \cmidrule{15-15} 
 & $\mathcal{E}_{pred}$ & $\mathcal{N}_{pred}$ &  & $\mathcal{E}_{pred}$ & $\mathcal{N}_{pred}$ &  & $\mathcal{E}_{pred}$ & $\mathcal{N}_{pred}$ &  & $\mathcal{E}_{pred}$ & $\mathcal{N}_{pred}$ &  & $\mathcal{E}_{pred}$ & $\mathcal{N}_{pred}$\tabularnewline
\midrule
$\left[{\rm Q}_{0},{\rm Q}_{1}\right]$ & 2.1E-8 & 4.2M &  & 1.7E-8 & 7.5M &  & 5.8E-7 & 11.7M &  & 2.7E-8 & 12.3M &  & 1.8E-7 & 15.3M\tabularnewline
$\left[{\rm Q}_{1},{\rm Q}_{2}\right]$ & 7.0E-8 & 9.4M &  & 2.2E-8 & 11.0M &  & 2.8E-7 & 12.6M &  & 2.2E-8 & 9.6M &  & 2.1E-7 & 12.8M\tabularnewline
$\left[{\rm Q}_{2},{\rm Q}_{3}\right]$ & 1.6E-7 & 11.9M &  & 1.1E-7 & 10.1M &  & 8.9E-8 & 17.5M &  & 2.9E-7 & 6.7M &  & 2.1E-7 & 12.6M\tabularnewline
$\left[{\rm Q}_{3},{\rm Q}_{4}\right]$ & 3.0E-7 & 12.6M &  & 4.2E-7 & 9.5M &  & 3.5E-8 & 18.5M &  & 4.2E-7 & 10.4M &  & 2.1E-7 & 12.1M\tabularnewline
$\left[{\rm Q}_{4},{\rm Q}_{5}\right]$ & 4.7E-7 & 13.9M &  & 5.4E-7 & 15.0M &  & 4.2E-8 & 6.6M &  & 3.1E-7 & 22.2M &  & 2.1E-7 & 12.1M\tabularnewline
\bottomrule
\end{tabular} \medskip{}
\par\end{centering}
\centering{}%
\begin{tabular}{ccc}
\toprule 
 & $\mathcal{E}_{pred}$ & $\mathcal{N}_{pred}$\tabularnewline
\midrule 
$\left[{\rm Q}_{0},{\rm Q}_{5}\right]$ & 2.03E-7 & 12.95M\tabularnewline
\bottomrule
\end{tabular}\caption{\label{tab:range_restriction} The table shows $\mathcal{E}_{pred}\left(\hat{\Gamma}\right)$
and $\mathcal{N}_{pred}\left(\hat{\Gamma},\mathcal{D}_{test}\right)$
when restricting the range for one of the inputs $\alpha_{0}$, $\nu$,
$\rho$, $K$, and $T$ into $\left[{\rm Q}_{k},{\rm Q}_{k+1}\right]$.
Here, $\left[{\rm Q}_{k},{\rm Q}_{k+1}\right]$ is a subset containing
the bottom $20k\%$ to $20(k+1)\%$ of a given set.}
\end{table}

We further investigate $\mathcal{E}_{pred}\left(\hat{\Gamma}\right)$
and $\mathcal{N}_{pred}\left(\hat{\Gamma},\mathcal{D}_{test}\right)$
while restricting the range of one of the inputs $\alpha_{0}$, $\nu$,
$\rho$, $K$, and $T$. Thus, let $\left[{\rm Q}_{k},{\rm Q}_{k+1}\right]$
be a subset containing the bottom $20k\%$ to $20(k+1)\%$ of a given
set. Recall that $\alpha_{0}$ is generated to be uniformly distributed
from $0$ to $2$ (if exactly speaking, from 0.01 to 2). Thus, $\left[{\rm Q}_{k},{\rm Q}_{k+1}\right]$
for $\alpha_{0}$ is $\{\left.\alpha_{0}\right|\frac{2}{5}k\leq\alpha_{0}\leq\frac{2}{5}\left(k+1\right)\}$.
Table \ref{tab:range_restriction} illustrates how $\mathcal{E}_{pred}\left(\hat{\Gamma}\right)$
and $\mathcal{N}_{pred}\left(\hat{\Gamma},\mathcal{D}_{test}\right)$
change when one of the inputs is restricted into $\left[{\rm Q}_{k},{\rm Q}_{k+1}\right]$.
In the lower table, $\left[{\rm Q}_{0},{\rm Q}_{5}\right]$ represents
that any restriction for the inputs is not imposed. Evidently, $\mathcal{N}_{pred}\left(\hat{\Gamma},\mathcal{D}_{test}\right)$
is large concerning specific domains (e.g., $\left[{\rm Q}_{4},{\rm Q}_{5}\right]$
for $\alpha_{0}$). The large $\mathcal{N}_{pred}\left(\hat{\Gamma},\mathcal{D}_{test}\right)$
means that $\mathcal{E}_{fit}\left(\hat{\Gamma},\mathcal{D}_{test}\right)/\mathcal{E}_{pred}\left(\hat{\Gamma}\right)$
is also high (see equation (\ref{eq:N_pred})). As $\mathcal{E}_{pred}\left(\hat{\Gamma}\right)$ 
positively correlates with $\mathcal{N}_{pred}\left(\hat{\Gamma},\mathcal{D}_{test}\right)$
in the table, this correlation leads to the fact that the domains with large $\mathcal{N}_{pred}\left(\hat{\Gamma},\mathcal{D}_{test}\right)$
have greater $\mathcal{E}_{fit}\left(\hat{\Gamma},\mathcal{D}_{test}\right)$.
Therefore, we conclude that the network reduces more noises in
those domains where $\mathcal{E}_{fit}\left(\hat{\Gamma},\mathcal{D}_{test}\right)$
is high. The phenomenon may occur because ordinary least squares (OLS)
is employed in this study. The MC produces larger standard deviations
$\beta_{l}$ on the specific domains. Although the targets $\sigma_{approx,l}^{I}$
for regression have different $\beta_{l}$, if OLS is adopted,
the capacity of the network can be exhausted due to $\sigma_{approx,l}^{I}$
with large deviations. Thus, the weighted least squares (WLS) should
be used to resolve it. However, estimating the deviations $\beta_{l}$
to utilize the WLS may be considerably difficult. 

\begin{figure}[t]
\begin{centering}
\includegraphics[scale=0.8]{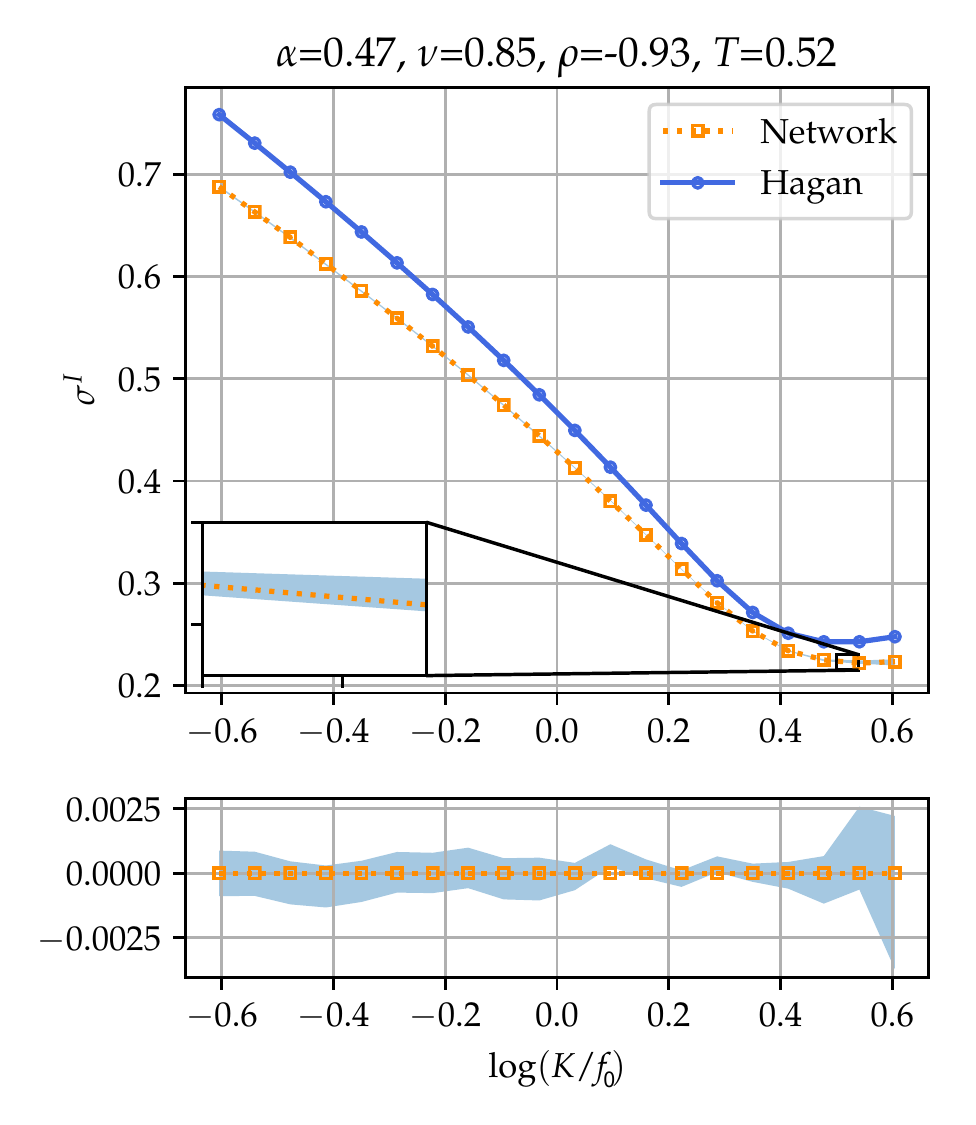} \includegraphics[scale=0.8]{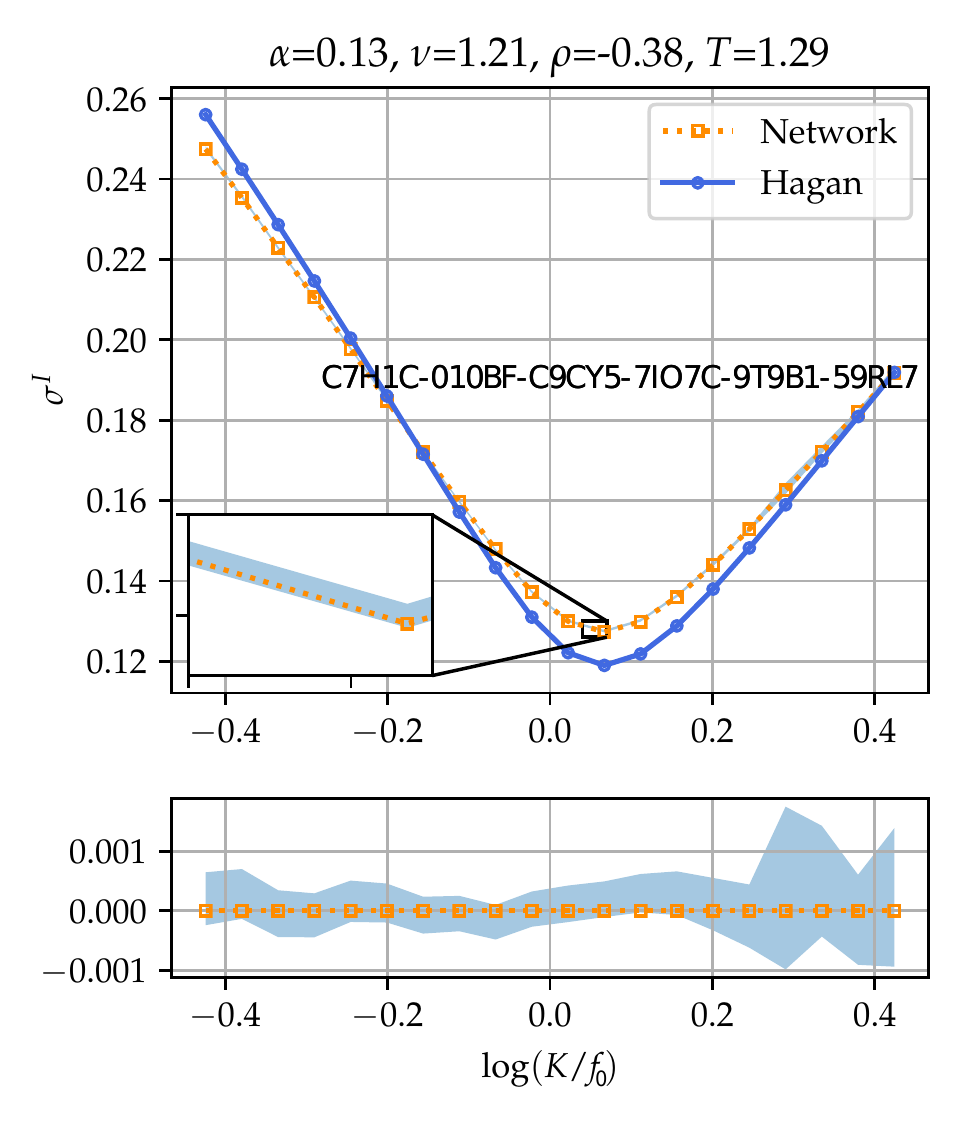}
\par\end{centering}
\caption{\label{fig:selection}The figures highlight the predictions of the best performance network for the data in the test set $\mathcal{D}_{test}$.
The blue regions indicate the 99\% confidence intervals computed by
10 million MC simulations. The intervals are drawn again centered
about $\sigma_{net}^{I}$ in the auxiliary subfigures at the bottom.}
\end{figure}
Finally, we reflect on the distribution of a random sample $\epsilon_{pred}$
from $\mathcal{D}_{pop}$, where $\mathcal{D}_{pop}$ is the population
producing $\epsilon_{pred,l}(=\sigma_{net,l}^{I}-\sigma_{true,l}^{I})$.
(Roughly, $\mathcal{D}_{pop}=\{\epsilon_{pred,l}\}_{l=1,\cdots,\infty}$
in that the population can be understood as the set of infinite samples.)
By doing so, we can guess how well $\{\sigma_{net,l}^{I}\}_{l=1,\cdots,L'}$
approximates $\{\sigma_{true,l}^{I}\}_{l=1,\cdots,L'}$. For instance,
we can talk about $P[|\sigma_{net,l}^{I}-\sigma_{true,l}^{I}|<1{\rm bp}]$.
First, $E[\epsilon_{pred}]=0$ because $E[\epsilon_{pred,l}]=0$ (see
relation (\ref{eq:err_dist})). Moreover, because $Var[\epsilon_{pred}]=E[\mathcal{E}_{pred}\left(\hat{\Gamma}\right)]$,
$Var[\epsilon_{pred}]$ is estimated as 2.03E-07 as depicted in Table \ref{tab:range_restriction}.
This association signifies that if the best performance network (i.e., 4 layers and 7,000 nodes per hidden layer) is utilized, ${\rm std}[\epsilon_{pred}]$ is about 0.00045 (0.45bp). Thus, under the assumption that $\epsilon_{pred}$ follows a normal distribution, approximately 97\% of $|\epsilon_{pred}|$ are within 1bp. This outcome obviously has remarkable accuracy. However, the normal assumption for the claim is hard to prove because estimating higher moments of $\epsilon_{pred}$ is difficult at the moment.
Alternatively, we try to draw many plots of $\sigma_{net}^{I}$
with respect to $\log\left(K/f_{0}\right)$, along with the 99\% confidence
intervals computed using 10 million MC simulations. We then confirm
that most of $\sigma_{net}^{I}$ are in the intervals. Owing to the
limitation of space, we select two and draw them in Figure \ref{fig:selection}.
In the figure, the blue regions indicate the 99\% confidence intervals.
The intervals are too thin to be observed, so auxiliary figures are
drawn at the bottom, in which the intervals are drawn centered about
$\sigma_{net}^{I}$. Notably, we can estimate the accuracy of the network through the subfigures. Based on the examination of the many plots, we speculate that the guess $|\epsilon_{pred}|<{\rm 1bp}$ with a high probability is not far wrong.

\section{Conclusion}

Recently, the application of deep learning algorithms has facilitated outstanding achievement in various fields. Considering that pricing options are truly essential in financial engineering, it seems inevitable to note
recent research using artificial neural networks to predict the
prices for particular parametric models. In our opinion, the models
without any pricing formulas should be studied more intensively. Nevertheless, measuring prediction errors is virtually impossible because true prices are unknown for such types of models. To resolve this problem, we develop a novel method based on the Monte-Carlo simulation and nonlinear regression. According to the method, the best performance network developed in this work produces the results as accurate as those of 13 million MC simulations. It is a remarkable result because the MC takes much more computational time in comparison to the network.

There are unresolved problems for future research. First, the decreasing
rate of the prediction errors does not exactly match with the value
our theory predicts. This discrepancy implies that the theory should be improved to address the gap. Second, the weighted least square should be introduced because the simulated option prices have different variances. Finally, only the expectation and variance of the prediction errors are provided in this work, but it is more desirable to investigate the theoretical distribution of the errors.

\section*{Acknowledgments}
We are grateful to Korean Asset Pricing, a bond rating agency located in Korea, for providing data in Figure 1. Jaegi Jeon received financial support from the National Research Foundation of Korea (NRF) of the Korean government (Grant No. NRF-2019R1I1A1A01062911). Jeonggyu Huh received financial support from the NRF (Grant No. NRF-2019R1F1A1058352).

\bibliographystyle{elsarticle-num-names}
\bibliography{ref}

\end{document}